\documentclass[11pt]{article}

\usepackage[T1]{fontenc}
\usepackage{amsmath,amssymb,amsthm}
\usepackage{array}
\usepackage{booktabs}
\usepackage{float}
\usepackage{graphicx}
\usepackage{microtype}
\usepackage[numbers,sort&compress]{natbib}
\usepackage{tabularx}
\usepackage[letterpaper,margin=1in]{geometry}
\usepackage{xcolor}
\usepackage{hyperref}
\usepackage{orcidlink}

\graphicspath{{figures/}}
\hypersetup{
  colorlinks=true,
  linkcolor=blue!55!black,
  citecolor=green!45!black,
  urlcolor=blue!70!black,
  pdftitle={Learning Aerosol Coagulation Dynamics in Latent Space with a Scale-Covariant Neural ODE},
}

\newtheorem{theorem}{Theorem}
\newtheorem{corollary}{Corollary}
\newtheorem{proposition}{Proposition}
\newtheorem{lemma}{Lemma}
\newcommand{\stopgrad}{\operatorname{stopgrad}}

\newcommand{\modelparam}[2]{#2}

\title{Learning Aerosol Coagulation Dynamics in Latent Space with a Scale-Covariant Neural ODE}
\author{%
Wenhan Tang$^{1}$\,\orcidlink{0009-0009-5076-036X},
Ruqi Yang$^{2,3}$\,\orcidlink{0000-0002-1602-7502},
Jeffrey H. Curtis$^{1}$\,\orcidlink{0000-0002-1447-2127},
Ehsan Saleh$^{4}$,\\
Lekha Patel$^{5}$\,\orcidlink{0000-0003-3508-0672},
Peter A. Bosler$^{5}$\,\orcidlink{0000-0002-3356-0296},
Nicole Riemer$^{1,*}$\,\orcidlink{0000-0002-3220-3457},
and Matthew West$^{4}$\,\orcidlink{0000-0002-7605-0050}\\[0.75em]
\parbox{0.95\textwidth}{\centering\small
$^{1}$Department of Climate, Meteorology \& Atmospheric Sciences, University of Illinois Urbana-Champaign, Urbana, IL, USA\\
$^{2}$Department of Forest and Wildlife Ecology, University of Wisconsin--Madison, Madison, WI, USA\\
$^{3}$Department of Computer Sciences, University of Wisconsin--Madison, Madison, WI, USA\\
$^{4}$Department of Mechanical Science and Engineering, University of Illinois Urbana-Champaign, Urbana, IL, USA\\
$^{5}$Center for Computing Research, Sandia National Laboratories, Albuquerque, NM, USA\\[0.75em]
\textsuperscript{*}Corresponding author: Nicole Riemer
(\href{mailto:nriemer@illinois.edu}{nriemer@illinois.edu})
}
}
\date{}

\begin{document}

\maketitle

\begin{abstract}
Particle-resolved aerosol models preserve the joint size--composition structure
that governs cloud activation, optical properties, and freezing, but their
computational cost limits their use in large-scale atmospheric models. We
extend AeroMELD from a compact representation of aerosol populations to a
prognostic model of coagulation. The resulting AeroMELD-Coag advances
\modelparam{model.latent_dim}{nine} learned coordinates for population shape
and one for total particle number through a scale-covariant neural ordinary
differential equation that incorporates the known concentration dependence of
coagulation. Across \modelparam{split.test_trajectories}{2{,}000} held-out
trajectories spanning 48 h, median symmetric errors are
\modelparam{results.latent.pooled_median_pct}{3.5\%} for latent shape and
\modelparam{results.number.pooled_median_pct}{0.52\%} for total number. This
\modelparam{model.total_state_dim}{ten}-coordinate state also retains the
evolving size-resolved composition and associated cloud-condensation-nuclei
(CCN) activation, optical properties, and frozen fraction. Compared with the
evaluated \modelparam{baseline.sectional.bin_count}{20}-bin sectional model
with \modelparam{baseline.sectional.state_dim}{320} prognostic coordinates,
AeroMELD-Coag lowers pooled mean CCN activation error from
\modelparam{results.sectional.ccn.pooled_mean_pct}{1.98\%} to
\modelparam{results.decoded.ccn.pooled_mean_pct}{1.38\%} and frozen-fraction
error from \modelparam{results.sectional.frozen.pooled_mean_pct}{1.72\%} to
\modelparam{results.decoded.frozen.pooled_mean_pct}{0.56\%}. In a matched
integration benchmark, GPU integration time per trajectory, amortized over a
batch, is more than three orders of magnitude lower than for both single-core
CPU references. These coagulation results provide a proof of concept for
advancing aerosol microphysics efficiently in a compact learned state while
retaining information about aerosol mixing state. They establish a foundation
for future large-scale atmospheric models to represent aerosol microphysics
with greater mixing-state detail within practical computational budgets.
\end{abstract}

\noindent\textbf{Keywords:} aerosol mixing state; coagulation; latent dynamics; particle-resolved modeling; scale covariance; multistep learning

\section{Introduction}\label{introduction}

Particle-resolved aerosol models retain the joint particle-size and composition
information needed to represent mixing-state effects on
cloud-condensation-nuclei (CCN) activity, heterogeneous ice nucleation, and
optical properties
\cite{poschl2005,mcfiggans2006,hoose2012,fierce2016,ching2017,riemer2019mixing,yao2022,tang2026freezing}.
This fidelity is computationally expensive and produces large,
variable-cardinality states, limiting repeated simulation and direct use across
atmospheric model grids. Conventional modal and sectional representations make
aerosol microphysics tractable by replacing variable particle populations with
prescribed modes or size bins
\cite{binkowski1995rpm,whitby1997modal,gelbard1980sectional,vignati2004m7,bauer2008matrix}.
These representations provide a practical reduced-order baseline for
large-scale models. The prognostic variables assigned to each mode or bin
determine the retained mixing-state information. A compact prognostic state
must support prognostic advancement and preserve the size--composition
associations that control the diagnostics of interest.

Machine learning offers complementary routes toward this goal by learning
aerosol representations and process evolution from high-fidelity simulations.
Combinatorial neural networks have learned coagulation on binned aerosol
states, while graph-network simulators have represented individual particles
as nodes to learn condensation-driven evolution
\cite{wang2022coagulation,ferracina2025gnn}. Particle-node graphs retain
individual-particle structure, but their prognostic state remains a
variable-cardinality node set; direct use in a conventional Eulerian grid
model would therefore require an additional strategy for transporting and
remapping those nodes between grid cells. For three-dimensional coupling,
compactness alone is not sufficient: the reduced state should also support
transport, air-mass mixing, and emission-source increments without rebuilding
a particle list in every cell. This motivates the complementary question
addressed here: can one fixed-dimensional learned state serve both as the
prognostic coordinate for nonlinear microphysics and as a carrier of multiple
mixing-state-sensitive diagnostics?

An initial variational Encoder--Decoder compressed binned,
size--composition-resolved aerosol states into ten latent variables while
reconstructing binned quantities and climate-relevant diagnostics
\cite{saleh2025generative}. A subsequent conditional generative framework used
partial aerosol observations to produce ensembles of plausible complete
states and uncertainty-aware diagnostic estimates
\cite{saleh2025partial}. AeroMELD then extended the learned representation to
weighted particle populations and introduced a scale--shape state comprising
total number and a \modelparam{model.latent_dim}{nine}-dimensional learned shape coordinate
\cite{saleh2026aeromeld}. Unlike a sectional state whose coordinates are fixed
by prescribed bins, AeroMELD learns its coordinates from the particle-level
joint distribution, and its Decoder is trained to retain size-resolved mass
and number, CCN activation, optical properties, and immersion-freezing
behavior. It learns a compact projection optimized for these selected outputs without recovering the identity of every particle.

Sectional and AeroMELD states allocate a fixed prognostic budget differently.
A sectional model assigns that budget to predefined size--composition bins.
AeroMELD learns the population directions that are most informative for the
retained outputs. The comparison evaluates whether those learned coordinates
can match a conventional reduced representation for coagulation with far fewer
prognostic variables.

AeroMELD's unnormalized learned moment
\(\mathbf m_\phi=N\mathbf z\) is linear in population measure. Consequently,
\(N\) and the nine components of \(\mathbf m_\phi\) form an additive,
fixed-dimensional state under population mixing and emission-source
injection, while the normalized shape \(\mathbf z\) can be recovered locally
for decoding and nonlinear process integration. This property provides a
direct tracer-like interface for future three-dimensional models and poses the
central question addressed here: is the same normalized state dynamically
sufficient as microphysical processes alter the particle population?

Binary coagulation provides a focused but nontrivial test of this question.
Pairwise collisions reduce particle number, combine particle mass and
composition, and thereby change the joint size--composition distribution,
mixing state, and associated diagnostics. Matching total number alone is
insufficient: a reduced model can reproduce bulk number decline while
misrepresenting the evolving size--composition associations that control cloud
activation, optical properties, and freezing behavior. At the same time, the
bilinear Smoluchowski collision operator fixes the concentration dependence of
the dynamics, allowing that known scaling to be embedded in the model
structure. This separation suggests an operator-split coupling strategy:
transport, mixing, and emissions act on the additive extensive coordinates,
whereas a local learned operator advances the nonlinear microphysical
evolution. Coagulation therefore tests whether the learned projection is
dynamically sufficient while separating known concentration physics from the
shape-dependent closure that must be learned. More broadly, it provides a
setting in which a scientific-machine-learning model can encode a known
operator homogeneity analytically and reserve learning for the unresolved
closure.

Here we present AeroMELD-Coag, which advances a
\modelparam{model.latent_dim}{nine}-dimensional latent shape and total number
with a scale-covariant neural ordinary differential equation trained through
diagnostic-aware recursive rollouts.
Across \modelparam{split.test_trajectories}{2{,}000} held-out 48-h
trajectories, AeroMELD-Coag attains median symmetric errors of
\modelparam{results.latent.pooled_median_pct}{3.5\%} for latent shape and
\modelparam{results.number.pooled_median_pct}{0.52\%} for total number. The
decoded rollouts retain coherent evolution in size-resolved quantities, CCN
activation, optical properties, and frozen fraction. The evaluated sectional baseline
provides a second reduced-order reference. With only
\modelparam{model.total_state_dim}{ten} prognostic coordinates, AeroMELD has
lower pooled mean errors for composition, CCN activation, and frozen fraction.
The sectional model has lower pooled mean errors for binned mass and number and
for the optical coefficients. In the matched integration benchmark, the
batch-amortized per-case L40S time was more than three orders of magnitude
below both single-core CPU references. The held-out rollouts and sectional
comparison therefore support a compact prognostic representation within the
evaluated fixed-environment regime.

\hypertarget{scale-covariant-aeromeld-dynamics}{%
\section{Scale-Covariant AeroMELD Dynamics}\label{scale-covariant-aeromeld-dynamics}}

\hypertarget{aerosol-population-and-scale-shape-representation}{%
\subsection{Aerosol population and scale--shape representation}\label{aerosol-population-and-scale-shape-representation}}

Table~\ref{tab:notation} summarizes the recurring notation used in the
AeroMELD-Coag formulation. Symbols used only in individual derivations or
loss definitions are defined locally where they appear.

\begin{table}[H]
\centering
\caption{Recurring notation and units used in the AeroMELD-Coag formulation.}
\label{tab:notation}
\footnotesize
\setlength{\tabcolsep}{4pt}
\renewcommand{\arraystretch}{1.04}
\begin{tabularx}{\textwidth}{@{}>{\raggedright\arraybackslash}p{0.20\textwidth}>{\raggedright\arraybackslash}X>{\raggedright\arraybackslash}p{0.18\textwidth}@{}}
\toprule
Symbol & Definition & Unit \\
\midrule
\multicolumn{3}{@{}l}{\textit{Population and representation}} \\
\addlinespace[2pt]
\(t\) &
Physical time. &
\(\mathrm{h}\) \\
\(p,\;N_p(t)\) &
Computational-particle index and number of computational particles at time \(t\). &
\(1\) \\
\(\mathcal X,\;x_{p,t}\) &
Particle composition-state space and the vector of \modelparam{data.species_count}{15} species masses carried by particle \(p\); here \(\mathcal X=\mathbb R_+^{\modelparam{data.species_count}{15}}\). &
Each component: \(\mathrm{kg}\) \\
\(w_{p,t},\;\mu_t\) &
Number concentration represented by particle \(p\), and the corresponding weighted population measure. &
\(\mathrm{m^{-3}}\) \\
\(N_t,\;\widehat\mu_t\) &
Total physical number concentration \(N_t=\mu_t(\mathcal X)\), and normalized particle-state distribution \(\widehat\mu_t=\mu_t/N_t\). For \(\widehat\mu_t\), the hat denotes normalization; elsewhere, a hat marks a model-produced quantity. &
\(N_t:\ \mathrm{m^{-3}};\quad \widehat\mu_t:\ 1\) \\
\(\boldsymbol\phi,\;\mathbf m_{\phi,t}\) &
Frozen learned particle feature map and its unnormalized population moment, \(\mathbf m_{\phi,t}=\int\boldsymbol\phi\,d\mu_t\). &
\(\boldsymbol\phi:\ 1;\quad \mathbf m_{\phi,t}:\ \mathrm{m^{-3}}\) \\
\(\mathbf z_t,\;d_z\) &
Latent-shape coordinate and its dimension; \(d_z=\modelparam{model.latent_dim}{9}\). &
\(1\) \\
\(N_{\mathrm{ref}},\;s_t\) &
Training-partition reference number concentration and dimensionless scale \(s_t=N_t/N_{\mathrm{ref}}\). &
\(N_{\mathrm{ref}}:\ \mathrm{m^{-3}};\quad s_t:\ 1\) \\
\addlinespace[2pt]
\multicolumn{3}{@{}l}{\textit{Decoding and physical outputs}} \\
\addlinespace[2pt]
\(D_u(\mathbf z),\;\mathbf u\) &
Frozen neural Decoder and its standardized, transformed model-space output. &
\(1\) \\
\(\mathcal D(\mathbf z,N),\;\mathbf q\) &
Complete physical decoding map and the physical-space output vector reconstructed or evaluated. &
Block-dependent \\
\(M_b,\;P_{ab},\;N_b\) &
Total mass concentration in size bin \(b\), fraction of species \(a\) in size bin \(b\), and number concentration in size bin \(b\). &
\(\mathrm{kg\,m^{-3}};\quad 1;\quad \mathrm{m^{-3}}\) \\
CCN-to-CN ratio; scattering and absorption; frozen fraction &
Remaining decoded diagnostic blocks. &
\(1;\quad \mathrm{m^{-1}};\quad 1\) \\
\addlinespace[2pt]
\multicolumn{3}{@{}l}{\textit{Coagulation dynamics and transformations}} \\
\addlinespace[2pt]
\(K(x,y),\;x\oplus y\) &
Symmetric binary-coagulation kernel and the particle state produced by merging states \(x\) and \(y\). &
\(K:\ \mathrm{m^3\,h^{-1}};\quad x\oplus y:\ \mathrm{kg}\) per component \\
\(F(\mathbf z),\;\mathbf G(\mathbf z)\) &
Physical scale- and latent-shape-tendency coefficient fields in \(ds/dt=s^2F(\mathbf z)\) and \(d\mathbf z/dt=s\mathbf G(\mathbf z)\). &
\(\mathrm{h^{-1}}\) \\
\(F_\theta(\mathbf z),\;\mathbf G_\theta(\mathbf z)\) &
Learned scale- and shape-rate coefficient heads in zero-centered transformed-rate coordinates. &
\(1\) \\
\(\tau,\;\Delta t\) &
Collision-exposure time defined by \(d\tau/dt=s\), and the numerical
integration step. &
\(\mathrm{h}\) \\
\(T_{\dot z},T_{\dot s};\;\overline T_{\dot z},\overline T_{\dot s}\) &
Fitted rate transformations and their zero-centered forms, for which zero physical rate maps to zero. &
\(1\) \\
\addlinespace[2pt]
\multicolumn{3}{@{}l}{\textit{Notation conventions}} \\
\addlinespace[2pt]
\(\dot{(\,\cdot\,)}\) &
Derivative with respect to physical time \(t\). &
Underlying unit multiplied by \(\mathrm{h^{-1}}\) \\
\((\,\cdot\,)^\star\) &
Reference quantity. Latent references are obtained from the frozen Encoder applied to the PartMC trajectory; physical references are particle-resolved values. &
Same as underlying quantity \\
\(\widehat{(\,\cdot\,)}\) &
Model-produced quantity, including a latent-dynamics prediction or Decoder reconstruction; \(\widehat\mu\) is the stated normalization exception. &
Same as underlying quantity \\
\(\widetilde{(\,\cdot\,)}\) &
Training quantity modified by the initial latent perturbation or its corresponding target correction. &
Same as underlying quantity \\
\(\theta,\;\operatorname{stopgrad}[\cdot]\) &
Trainable latent-dynamics parameters, and the forward-pass identity whose argument is detached from gradient propagation. &
\(\theta:\ 1;\quad \operatorname{stopgrad}[\cdot]\): same as argument \\
\bottomrule
\end{tabularx}
\end{table}

At time \(t\), a particle-resolved aerosol population is represented by a weighted measure
\begin{equation}
\begin{aligned}
\mu_t&=\sum_{p=1}^{N_p(t)}w_{p,t}\,\delta_{x_{p,t}},\\
N_t&:=\mu_t(\mathcal X)=\sum_{p=1}^{N_p(t)}w_{p,t},\\
\widehat\mu_t&:=\frac{\mu_t}{N_t}.
\end{aligned}
\label{eq:population-measure}
\end{equation}
For the present data, \(\mathcal X=\mathbb R_+^{\modelparam{data.species_count}{15}}\), and the particle state is
\(x_{p,t}=(m_{p,t,1},\ldots,m_{p,t,\modelparam{data.species_count}{15}})\in\mathcal X\),
where the components are the masses of the \modelparam{data.species_count}{15} aerosol species carried by
computational particle \(p\). The raw population state is therefore the
variable-cardinality pair
\(X_t\in\mathbb R_+^{N_p(t)\times\modelparam{data.species_count}{15}}\) and
\(\mathbf w_t\in\mathbb R_+^{N_p(t)}\), where row \(p\) of \(X_t\) is
\(x_{p,t}\) and \(w_{p,t}\) is the represented number concentration of that
computational particle. The formulation uses the weighted measure in Eq.~\eqref{eq:population-measure} as the mathematical population object instead of flattening it into a fixed \(15N_p(t)\)-component vector.
Here \(N_p(t)\) is the variable number of computational particles, \(N_t\)
is the total physical number concentration, and \(\widehat\mu_t\) is the
normalized particle-state distribution whenever \(N_t>0\).

The frozen AeroMELD Encoder applies a learned particle feature map \(\boldsymbol\phi:\mathcal X\rightarrow\mathbb R^{d_z}\) and aggregates those features with the normalized particle weights \cite{saleh2026aeromeld,zaheer2017deepsets}:
\begin{equation}
\begin{aligned}
\mathbf z_t
&=\int_{\mathcal X}\boldsymbol\phi(x)\,d\widehat\mu_t(x),\\
\mathbf m_{\phi,t}
&:=\int_{\mathcal X}\boldsymbol\phi(x)\,d\mu_t(x)
=N_t\mathbf z_t,
\qquad d_z=\modelparam{model.latent_dim}{9}.
\end{aligned}
\label{eq:latent-shape}
\end{equation}
We index the components of \(\mathbf z_t\) as \(z_{j,t}\), with
\(j=1,\ldots,d_z\); figure legends abbreviate these components as
\(z_1,\ldots,z_9\).
The particle feature map may be nonlinear in \(x\), but the unnormalized moment \(\mathbf m_\phi\) is linear in population measure:
\begin{equation}
\mathbf m_\phi(\alpha\mu_1+\beta\mu_2)
=\alpha\mathbf m_\phi(\mu_1)+\beta\mathbf m_\phi(\mu_2),
\qquad \alpha,\beta\ge0.
\label{eq:linear-population-moment}
\end{equation}
The normalized latent shape \(\mathbf z=\mathbf m_\phi/N\) is therefore a barycenter of learned particle features, whereas \(N\) carries population scale. The complete reduced aerosol state used here has \modelparam{model.total_state_dim}{ten} coordinates: \(N\) and the \modelparam{model.latent_dim}{nine} components of \(\mathbf z\).

The frozen Decoder's neural map
\(D_u:\mathbb R^{\modelparam{model.latent_dim}{9}}\rightarrow\mathbb R^{\modelparam{decoder.output_dim}{550}}\) sends \(\mathbf z\) to a
transformed model-space vector \(\mathbf u\) with \modelparam{decoder.block_count}{seven} fixed ordered blocks:
total binned mass \(M_b\in\mathbb R^{\modelparam{decoder.mass_bins}{20}}\), normalized size-resolved
composition \(P_{ab}\in\mathbb R^{\modelparam{decoder.mass_bins}{20}\times\modelparam{decoder.composition_species}{15}}\), binned number
\(N_b\in\mathbb R^{\modelparam{decoder.mass_bins}{20}}\), CCN activation in \(\mathbb R^{\modelparam{decoder.ccn_points}{100}}\), scattering
and absorption in \(\mathbb R^{\modelparam{decoder.optical_points}{5}}\) each, and frozen fraction in
\(\mathbb R^{\modelparam{decoder.frozen_points}{100}}\). These dimensions sum to
\(\modelparam{decoder.mass_bins}{20}+\modelparam{decoder.mass_bins}{20}\times\modelparam{decoder.composition_species}{15}+\modelparam{decoder.mass_bins}{20}+\modelparam{decoder.ccn_points}{100}+\modelparam{decoder.optical_points}{5}+\modelparam{decoder.optical_points}{5}+\modelparam{decoder.frozen_points}{100}=\modelparam{decoder.output_dim}{550}\). Block-specific inverse
transformations, using \(N\) for the scale-dependent blocks, give the complete
physical reconstruction
\begin{equation}
\widehat{\mathbf q}_t=\mathcal D(\mathbf z_t,N_t).
\label{eq:diagnostic-decoder}
\end{equation}
Here \(\widehat{\mathbf q}_t\) denotes a physical-space output reconstructed by
the complete Decoder. The hat marks a model-produced reconstruction whose input \((\mathbf z_t,N_t)\) may be either an encoded reference state or a predicted rollout state.

The Encoder and Decoder define the representation and remain fixed in this
study. The new learning problem is to advance \((\mathbf z,N)\) under
coagulation and then use the same Decoder to diagnose the evolving particle
population.

The linearity of the unnormalized learned moment characterizes the representation and exposes population scale explicitly. Multiplying every
particle weight by the same positive factor changes \(N\) and
\(\mathbf m_\phi\) by that factor while leaving the normalized distribution
\(\widehat\mu\) and latent shape \(\mathbf z\) unchanged. The resulting scale--shape coordinates allow the known concentration homogeneity of binary coagulation to be imposed analytically, so the neural dynamics model can focus on the remaining shape dependence.

\hypertarget{scale-covariant-coagulation-dynamics}{%
\subsection{Scale-covariant coagulation dynamics}\label{scale-covariant-coagulation-dynamics}}

We express population scale with the dimensionless coordinate
\begin{equation}
s_t=\frac{N_t}{N_{\mathrm{ref}}},
\label{eq:number-scale}
\end{equation}
where \(N_{\mathrm{ref}}\) is the arithmetic mean of total number concentration
over valid training states.

Before deriving the dynamics, consider a pure concentration rescaling in which
every computational-particle weight is multiplied by the same factor \(c>0\):
\begin{equation}
\mu^{(c)}=c\mu,
\qquad
N^{(c)}=cN,
\qquad
\widehat\mu^{(c)}=\widehat\mu,
\qquad
\mathbf z^{(c)}=\mathbf z.
\label{eq:concentration-rescaling}
\end{equation}
The two population states therefore have the same normalized particle-state
distribution and the same latent shape, but different number concentrations.
Binary coagulation is bilinear in the population measure because every event
is formed from a particle pair. Consequently, unnormalized population
tendencies acquire two powers of concentration. The normalized latent-shape
tendency has a different scaling because dividing the learned moment by total
number removes one of those powers.

For binary coagulation, let \(K(x,y)\ge0\) be the symmetric collision kernel and
let \(x\oplus y\) denote the merged particle state. For any test function
\(\psi:\mathcal X\rightarrow\mathbb R\), the weak Smoluchowski equation is
\begin{equation}
\begin{aligned}
\frac{d}{dt}\int_{\mathcal X}\psi(x)\,d\mu_t(x)
&=\frac12\iint_{\mathcal X\times\mathcal X}K(x,y)
\left[\psi(x\oplus y)-\psi(x)-\psi(y)\right]
d\mu_t(x)d\mu_t(y).
\end{aligned}
\label{eq:smoluchowski-weak}
\end{equation}
The bracket records one merged-particle gain and two parent-particle losses.
Choosing \(\psi=1\) isolates total number, whereas choosing
\(\psi=\boldsymbol\phi\) componentwise evolves the unnormalized learned moment.
The following theorem shows explicitly how the two tendencies combine in the
normalized latent coordinate.

\begin{theorem}[Scale-covariant latent dynamics of binary coagulation]
\label{thm:scale-covariant-dynamics}
Let a population evolve according to Eq.~\eqref{eq:smoluchowski-weak}, with
\(N\), \(\widehat\mu\), \(\mathbf m_\phi\), and \(\mathbf z\) defined in
Section~\ref{aerosol-population-and-scale-shape-representation}. Define the
normalized-population functionals
\begin{equation}
\begin{aligned}
 a(\widehat\mu)
&:=-\frac12\iint K(x,y)\,
d\widehat\mu(x)d\widehat\mu(y),\\
\mathbf b(\widehat\mu)
&:=\frac12\iint K(x,y)
\left[
\boldsymbol\phi(x\oplus y)
-\boldsymbol\phi(x)
-\boldsymbol\phi(y)
\right]
d\widehat\mu(x)d\widehat\mu(y),\\
\mathbf g(\widehat\mu)
&:=\mathbf b(\widehat\mu)-\mathbf z\,a(\widehat\mu).
\end{aligned}
\label{eq:normalized-tendency-functionals}
\end{equation}
Then
\begin{equation}
\frac{dN}{dt}=N^2a(\widehat\mu),
\qquad
\frac{d\mathbf m_\phi}{dt}=N^2\mathbf b(\widehat\mu),
\qquad
\frac{d\mathbf z}{dt}=N\mathbf g(\widehat\mu).
\label{eq:scale-shape-functionals}
\end{equation}
If the normalized-population tendencies close through the latent coordinate on
the modeled manifold, so that
\(a(\widehat\mu)=\widetilde a(\mathbf z)\) and
\(\mathbf g(\widehat\mu)=\widetilde{\mathbf g}(\mathbf z)\), then
\begin{equation}
\boxed{\frac{ds}{dt}=s^2F(\mathbf z)},
\qquad
\boxed{\frac{d\mathbf z}{dt}=s\mathbf G(\mathbf z)},
\label{eq:scale-covariant-dynamics}
\end{equation}
where
\(F(\mathbf z)=N_{\mathrm{ref}}\widetilde a(\mathbf z)\) and
\(\mathbf G(\mathbf z)=N_{\mathrm{ref}}\widetilde{\mathbf g}(\mathbf z)\).
\end{theorem}

\begin{proof}
\textit{Total-number tendency.}
Setting \(\psi=1\) in Eq.~\eqref{eq:smoluchowski-weak} makes the bracket equal
to \(-1\). Using \(\mu=N\widehat\mu\) in both integration factors gives
\[
\frac{dN}{dt}
=-\frac12\iint K(x,y)\,d\mu(x)d\mu(y)
=-\frac{N^2}{2}\iint K(x,y)\,
d\widehat\mu(x)d\widehat\mu(y)
=N^2a(\widehat\mu).
\]

\textit{Unnormalized learned-moment tendency.}
Setting \(\psi=\boldsymbol\phi\) componentwise and again substituting
\(\mu=N\widehat\mu\) gives
\[
\frac{d\mathbf m_\phi}{dt}
=\frac{N^2}{2}\iint K(x,y)
\left[
\boldsymbol\phi(x\oplus y)
-\boldsymbol\phi(x)
-\boldsymbol\phi(y)
\right]
d\widehat\mu(x)d\widehat\mu(y)
=N^2\mathbf b(\widehat\mu).
\]

\textit{Normalized latent-shape tendency.}
Because \(\mathbf z=\mathbf m_\phi/N\), the quotient rule combines the two
unnormalized tendencies:
\[
\frac{d\mathbf z}{dt}
=\frac1N\frac{d\mathbf m_\phi}{dt}
-\frac{\mathbf m_\phi}{N^2}\frac{dN}{dt}
=N\left[
\mathbf b(\widehat\mu)-\mathbf z\,a(\widehat\mu)
\right]
=N\mathbf g(\widehat\mu).
\]
Thus one concentration factor remains in the physical-time shape tendency
after normalization removes the other.

\textit{Closed scale--shape dynamics.}
Under closure through \(\mathbf z\), substitute \(N=N_{\mathrm{ref}}s\) into
the number and latent-shape tendencies. Dividing the number equation by
\(N_{\mathrm{ref}}\) and using the definitions of \(F\) and \(\mathbf G\)
gives Eq.~\eqref{eq:scale-covariant-dynamics}.
\end{proof}

The derivation separates exact structure from learned approximation. The
quadratic scale tendency and linear physical-time shape speed follow exactly
from the bilinear collision operator and the normalized scale--shape
coordinates. Closure through the nine-dimensional \(\mathbf z\) remains the reduced-model assumption tested by the held-out rollouts. If the
collision kernel depends on a time-varying environment, the same concentration
powers remain valid, but the closed coefficient fields generally also depend
on that environment.

At two corresponding population states related by the rescaling in
Eq.~\eqref{eq:concentration-rescaling}, the local vector field therefore
transforms as
\begin{equation}
N\mapsto cN,
\qquad
\mathbf z\mapsto\mathbf z,
\qquad
\frac{dN}{dt}\mapsto c^2\frac{dN}{dt},
\qquad
\frac{d\mathbf z}{dt}\mapsto c\frac{d\mathbf z}{dt}.
\label{eq:local-scale-covariance}
\end{equation}
The result is covariant rather than invariant: concentration rescaling leaves
the latent-shape coordinate unchanged but changes its physical-time velocity
by a prescribed factor.

The covariance structure determines what the neural model must learn. An unrestricted
neural ordinary differential equation taking \((\mathbf z,s)\) as input would
have to infer both the known concentration homogeneity and the unknown shape
dependence from finite training data. Such a model could interpolate the sampled concentration range yet violate the required behavior at substantially different scales. The structured form supplies the exact factors \(s^2\) and \(s\) analytically and restricts the network to learning the remaining functions of latent shape. The concentration-scaling form is therefore
enforced for every positive \(s\), including scales not represented during
training, provided that the binary-coagulation assumptions remain applicable.
Accuracy for unseen latent shapes, environmental conditions, or processes remains an empirical question beyond this structural extrapolation.

The same scientific-machine-learning design principle operates at both stages
of AeroMELD. The Encoder treats the variable-cardinality \(N_p\times15\) particle matrix as a set-structured input. Instead, a shared particle feature map and normalized
number-weighted aggregation encode permutation invariance and expose a
population moment that is linear in the measure. At the dynamics stage, the model encodes the known binary-collision scale dependence and learns the unresolved shape-dependent closure.

\hypertarget{latent-coagulation-paths}{%
\subsection{Latent coagulation paths}\label{latent-coagulation-paths}}

We write \(\dot{\mathbf z}:=d\mathbf z/dt\) and \(\dot s:=ds/dt\). Equation~\eqref{eq:scale-covariant-dynamics} separates the local direction of latent-shape evolution from the concentration-dependent speed.

\begin{corollary}[Concentration reparameterizes latent coagulation trajectories]
\label{cor:latent-direction}
At a state with \(s>0\) and \(\mathbf G(\mathbf z)\ne\mathbf0\),
\begin{equation}
\frac{\dot{\mathbf z}}{\|\dot{\mathbf z}\|_2}
=\frac{\mathbf G(\mathbf z)}{\|\mathbf G(\mathbf z)\|_2}.
\label{eq:latent-direction}
\end{equation}
Defining collision-exposure time by
\begin{equation}
\tau(t)=\int_0^t s(u)\,du
\label{eq:collision-exposure}
\end{equation}
gives
\begin{equation}
\frac{d\mathbf z}{d\tau}=\mathbf G(\mathbf z).
\label{eq:exposure-dynamics}
\end{equation}
Consequently, where the autonomous reduced equation has a unique solution, populations that start from the same latent state \(\mathbf z_0\) follow the same latent path in collision-exposure time; positive concentration changes the mapping between that path and physical time.
\end{corollary}

\begin{proof}
For \(s>0\) and \(\mathbf G(\mathbf z)\ne\mathbf0\), substituting
\(\dot{\mathbf z}=s\mathbf G(\mathbf z)\) into the normalized tangent cancels
the positive scalar \(s\), giving Eq.~\eqref{eq:latent-direction}. Because
\(d\tau/dt=s>0\), the chain rule then gives
\(d\mathbf z/d\tau=\mathbf G(\mathbf z)\).
\end{proof}

Concentration continues to control physical-time evolution by setting how rapidly the latent path is traversed. After reparameterization, the path-generating field \(\mathbf G(\mathbf z)\) is concentration-independent, whereas \(d\mathbf z/dt\) retains its concentration dependence.

The event-level form of the same field is
\begin{equation}
\begin{aligned}
\frac{d\mathbf z}{dt}
&=\frac{N}{2}\iint K(x,y)
\left[\boldsymbol\phi(x\oplus y)-\boldsymbol\phi(x)-\boldsymbol\phi(y)+\mathbf z\right]
d\widehat\mu(x)d\widehat\mu(y).
\end{aligned}
\label{eq:event-latent-tangent}
\end{equation}
The \(+\mathbf z\) term is the quotient-rule correction produced when the unnormalized moment \(\mathbf m_\phi=N\mathbf z\) is converted to the normalized barycenter. Thus \(\mathbf G(\mathbf z)=d\mathbf z/d\tau\) is a coordinate-dependent latent tendency that aggregates coagulation events in one fixed learned feature basis.

Two supporting properties are recorded in Appendix~\ref{supporting-scale-shape-results}. While collisions remain active, strictly decreasing total number prevents recurrence of the complete state \((N,\mathbf z)\). Under Lipschitz particle features, Wasserstein continuity of the normalized population path transfers to continuity of \(\mathbf z(t)\). These properties support a regular latent trajectory. Small curvature, injectivity, and exact closure remain unresolved.

The corollary motivates the architecture in Section~\ref{structured-latent-neural-ode}: the known factors \(s\) and \(s^2\) enter analytically, leaving the neural network to learn only the shape-dependent coefficients.

\hypertarget{structured-latent-neural-ode}{%
\subsection{Structured latent neural ODE}\label{structured-latent-neural-ode}}

The latent-dynamics network receives \(\mathbf z\in\mathbb R^{\modelparam{model.latent_dim}{9}}\) and uses \modelparam{dynamics.hidden_layers}{two} fully connected hidden layers of width \modelparam{dynamics.hidden_width}{64} with ReLU activations and dropout probability \modelparam{dynamics.dropout}{0.1}. A joint \modelparam{model.total_state_dim}{ten}-component affine output layer is split into a \modelparam{model.latent_dim}{nine}-component shape-rate coefficient \(\mathbf G_\theta(\mathbf z)\) and a scalar scale-rate coefficient \(F_\theta(\mathbf z)\). Including biases, the network has \modelparam{dynamics.parameter_count}{5{,}450} trainable parameters. Applying the analytic scale factors to these outputs defines a structured autonomous neural ordinary differential equation for \((\mathbf z,s)\) \cite{chen2018neuralode}. The updated scale--shape state is the complete learned rollout state; no additional learned recurrent memory is carried between steps.

Adapting AeroMELD's preprocessing and post-processing strategy of fitted power transformations and standardization \cite{saleh2026aeromeld}, the continuous-time physical rates are compared and predicted through fitted signed-power-plus-standardization transformations. For a scalar or vector rate \(r\), define the elementwise signed-power map \(S_\alpha(r)=\operatorname{sign}(r)|r|^\alpha\). The transformations
\begin{equation}
\begin{aligned}
T_{\dot z}(\dot{\mathbf z})
&=\frac{S_{\alpha_{\dot z}}(\dot{\mathbf z})-\boldsymbol\mu_{\dot z}}
{\boldsymbol\sigma_{\dot z}},\\
T_{\dot s}(\dot s)
&=\frac{S_{\alpha_{\dot s}}(\dot s)-\mu_{\dot s}}
{\sigma_{\dot s}},
\qquad \alpha_{\dot z}=\alpha_{\dot s}=\modelparam{transform.latent_rate_power}{1}.
\end{aligned}
\label{eq:rate-transformations}
\end{equation}
use componentwise location and scale statistics fitted on the training partition. Both rate families use hours as the physical time unit, and the same fitted parameters are used for training, evaluation, and inverse recovery. We remove the transformed coordinate of zero physical rate by defining
\begin{equation}
\begin{aligned}
\overline T_{\dot z}(\dot{\mathbf z})
&:=T_{\dot z}(\dot{\mathbf z})-T_{\dot z}(\mathbf0),
&
\overline T_{\dot z}^{-1}(\mathbf u)
&:=T_{\dot z}^{-1}\!\left[\mathbf u+T_{\dot z}(\mathbf0)\right],\\
\overline T_{\dot s}(\dot s)
&:=T_{\dot s}(\dot s)-T_{\dot s}(0),
&
\overline T_{\dot s}^{-1}(u)
&:=T_{\dot s}^{-1}\!\left[u+T_{\dot s}(0)\right].
\end{aligned}
\label{eq:centered-rate-transformations}
\end{equation}

Because the selected exponents satisfy
\(\alpha_{\dot z}=\alpha_{\dot s}=\modelparam{transform.latent_rate_power}{1}\), the signed-power maps reduce to the
identity. Subtracting the transformed value of zero therefore gives
\[
\overline T_{\dot z}(\dot{\mathbf z})
=\dot{\mathbf z}/\boldsymbol\sigma_{\dot z},
\qquad
\overline T_{\dot s}(\dot s)
=\dot s/\sigma_{\dot s},
\]
with componentwise division for the latent rate. These zero-centered
transformations are linear maps through the origin, so multiplication by
\(s\) or \(s^2\) commutes with the transformation. The analytic factors in
Eq.~\eqref{eq:structured-rate-heads} therefore preserve the physical
concentration homogeneity in the transformed coordinates used by the network.

At a predicted state \((\widehat{\mathbf z},\widehat s)\), the joint network output is split to return the zero-centered transformed-rate coordinates,
\begin{equation}
\begin{aligned}
\widehat{\mathbf u}_{\dot z}(\widehat s,\widehat{\mathbf z})
&=\widehat s\,\mathbf G_\theta(\widehat{\mathbf z}),
&
\widehat u_{\dot s}(\widehat s,\widehat{\mathbf z})
&=\widehat s^2F_\theta(\widehat{\mathbf z}),\\
\widehat{\dot{\mathbf z}}
&=\overline T_{\dot z}^{-1}(\widehat{\mathbf u}_{\dot z}),
&
\widehat{\dot s}
&=\overline T_{\dot s}^{-1}(\widehat u_{\dot s}).
\end{aligned}
\label{eq:structured-rate-heads}
\end{equation}
The first row is the direct model output in transformed coordinates; the second row restores the physical rates supplied to the integrator.

For training and default evaluation, we integrate this neural ODE with a
fixed-step forward-Euler scheme. This first-order, single-stage update requires
one learned right-hand-side evaluation per step. The implementation persists
physical total number and evaluates
\(\widehat s_k=\widehat N_k/N_{\mathrm{ref}}\) at each step:
\begin{equation}
\begin{aligned}
\widehat{\mathbf z}_{k+1}
&=\widehat{\mathbf z}_k+\Delta t\,\widehat{\dot{\mathbf z}}_k,\\
\widehat N_{k+1}
&=\widehat N_k+\Delta t\,N_{\mathrm{ref}}\widehat{\dot s}_k.
\end{aligned}
\label{eq:euler-step}
\end{equation}
The learned right-hand side returns continuous-time physical rates independently of the discretization; the integration formula and \(\Delta t\) enter only through the numerical solver. Section~\ref{integration-step-robustness}
tests fixed-step Euler integration from 1 s to 8 h and, at the deliberately
coarse \(\Delta t=8\) h stress test, replaces Euler with classical
fourth-order Runge--Kutta (RK4) without retraining. RK4 provides a higher-order comparison at the same coarse step and uses four right-hand-side evaluations per step, compared with one for Euler.

Figure~\ref{fig:method_overview} summarizes the recursive
encode--advance--decode pathway and expands one structured neural-ODE update.
The inset expresses concentration through \(s=N/N_{\mathrm{ref}}\); the
implementation persists physical total number through the equivalent update
in Eq.~\eqref{eq:euler-step}.

\begin{figure}[!htbp]
\centering
\includegraphics[width=\textwidth]{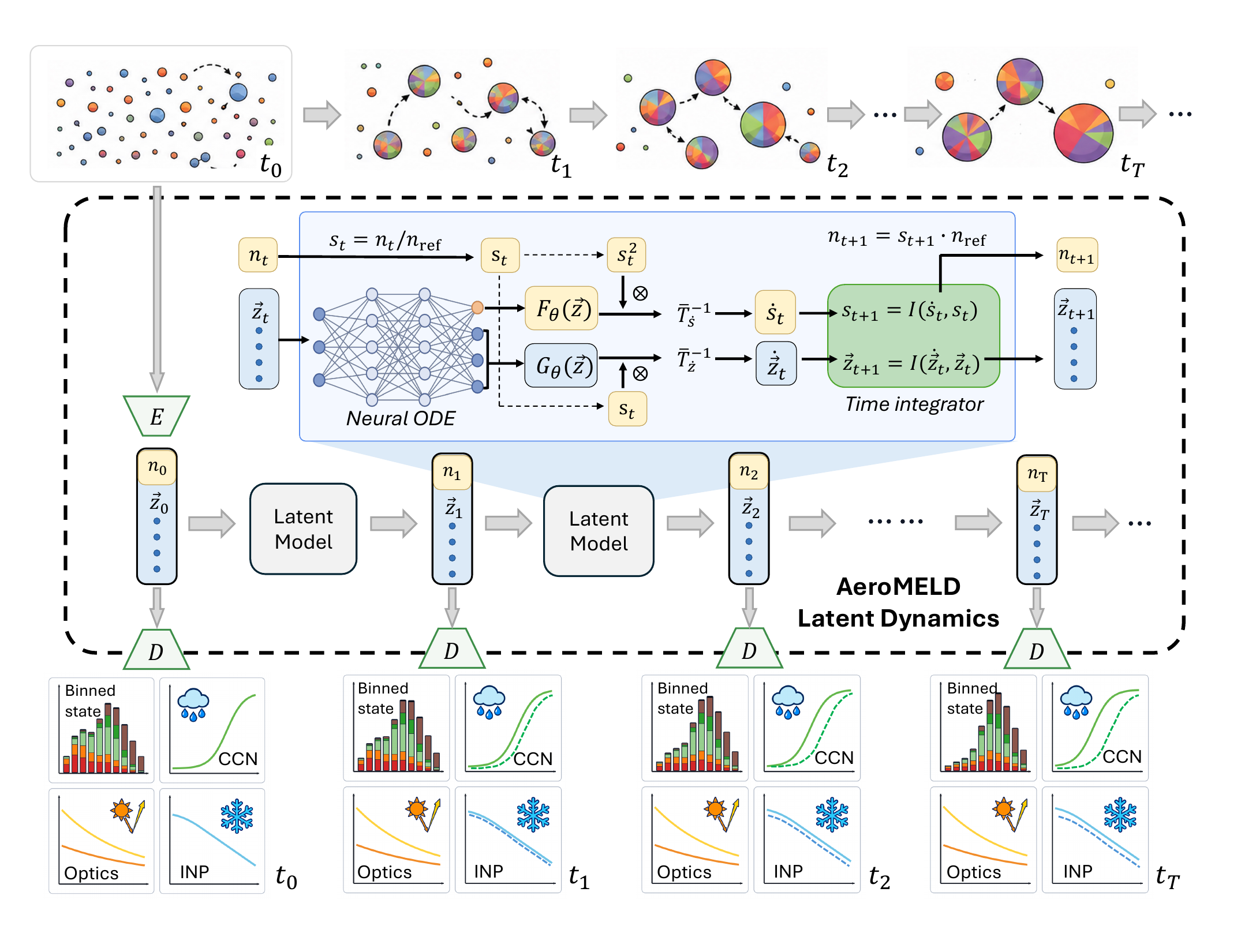}
\caption{Architecture and rollout workflow of AeroMELD-Coag. The upper
sequence schematically illustrates aerosol-population evolution by
coagulation. The middle dashed box contains the core framework for AeroMELD
latent dynamics: the initial particle population supplies total number
\(N_0\) and is mapped by the frozen Encoder to latent shape \(\mathbf z_0\),
after which the latent model recursively advances \((N_t,\mathbf z_t)\). The
lower sequence shows the binned aerosol state, the
cloud-condensation-nuclei (CCN) activation spectrum, the spectra of the volume
scattering and absorption coefficients, and the frozen-fraction spectrum
decoded from \((N_t,\mathbf z_t)\) at each output time.
In the diagram, \(E\) and \(D\) denote the frozen AeroMELD Encoder and Decoder
based on the framework detailed by Saleh~et~al.~(2026)~
\cite{saleh2026aeromeld}. The inset shows one
scale-covariant update: with \(s_t=N_t/N_{\mathrm{ref}}\), the analytic factors
\(s_t\) and \(s_t^2\) impose the concentration dependence on
\(\mathbf G_\theta(\mathbf z_t)\) and \(F_\theta(\mathbf z_t)\), respectively,
before inverse rate transformation and numerical integration.}
\label{fig:method_overview}
\end{figure}

\hypertarget{frozen-decoder-reference-targets}{%
\subsection{Frozen-decoder reference targets}\label{frozen-decoder-reference-targets}}

The frozen Encoder--Decoder defines a fixed representation manifold, and the latent-dynamics model is trained to follow the reference trajectory on that manifold. A target formed directly from raw physical outputs retains the frozen Decoder's static reconstruction residual and can therefore pressure the dynamics away from the reference latent endpoint. Rate supervision alone, by contrast, supplies no direct weighting for latent directions that strongly affect the aerosol outputs retained by the Decoder. The frozen-decoder reference target (FDT) addresses both roles by comparing the predicted endpoint with the image of the unperturbed reference endpoint under the same frozen Decoder, while detaching the reference branch.

Throughout the training equations, a superscript \(\star\) marks a reference
quantity: latent states, scales, and rates are derived from the frozen Encoder
applied to the PartMC trajectory, whereas physical reference outputs are taken
from the particle-resolved trajectory. The superscript marks a reference quantity; \(\operatorname{stopgrad}\) separately indicates detachment of the FDT reference branch.

Let \(D_u(\mathbf z)\in\mathbb R^{\modelparam{decoder.output_dim}{550}}\) denote the output of the frozen neural Decoder in its transformed model space. It contains \modelparam{decoder.block_count}{seven} blocks: binned mass, size-resolved composition, binned number, cloud-condensation-nuclei activation, scattering, absorption, and frozen fraction. Physical inverse transformations subsequently use \(N\) for the scale-dependent blocks. Writing \(D_{u,q}\) for block \(q\) with dimension \(d_q\) and \(\mathcal Q\) for the \modelparam{decoder.block_count}{seven} blocks, the per-window endpoint loss is
\begin{equation}
\ell_{\mathrm{FDT},i,k}^{(h)}
=\frac{1}{|\mathcal Q|}\sum_{q\in\mathcal Q}
\frac{1}{d_q}
\left\|
D_{u,q}(\widehat{\mathbf z}_{i,k+h})
-\stopgrad\!\left[D_{u,q}(\mathbf z^\star_{i,k+h})\right]
\right\|_2^2.
\label{eq:frozen-decoder-target}
\end{equation}
The predicted branch remains differentiable through the frozen Decoder to the latent-dynamics network. Because both branches use the same fixed map, a prediction that reaches the reference latent endpoint incurs zero FDT loss independently of the frozen representation residual; away from that endpoint, the loss weights latent errors by their effect on retained binned and diagnostic outputs. Appendix~\ref{frozen-decoder-target-properties} proves this residual-neutrality property and records its local Jacobian interpretation.

\hypertarget{multistep-rollout-learning}{%
\subsection{Multistep rollout learning}\label{multistep-rollout-learning}}

Rollout-aware objectives expose a learned vector field to its own evolving states and have been used to improve prognostic stability in weather and climate parameterization, for example by Brenowitz and Bretherton \cite{brenowitz2018prognostic}. For trajectory \(i\), start index \(k\), and horizon \(h\in\mathcal H=\{\modelparam{training.horizon_min_h}{1},\ldots,\modelparam{training.horizon_max_h}{12}\}\), define the consecutive reference window
\begin{equation}
\mathcal W_{i,k}^{(h)}
=\left((\mathbf z^\star_{i,k+j},s^\star_{i,k+j})\right)_{j=0}^{h}.
\label{eq:rollout-window}
\end{equation}
Adjacent states are separated by \(\Delta t=\modelparam{training.integration_step_h}{1}\) h. Windows remain within one trajectory, may overlap in time, and use every eligible start on the hourly grid.

During training, each rollout receives one initial perturbation. We draw
\begin{equation}
\boldsymbol\xi_{i,k}\sim\mathcal N(\mathbf0,\modelparam{training.perturbation_scale}{0.1}^2I_{d_z}),
\qquad
\boldsymbol\epsilon_{i,k}
=\Delta t\,\overline T_{\dot z}^{-1}(\boldsymbol\xi_{i,k}),
\qquad
\widetilde{\mathbf z}_{i,k}
=\mathbf z^\star_{i,k}+\boldsymbol\epsilon_{i,k},
\label{eq:training-perturbation}
\end{equation}
where \(\boldsymbol\xi\) is expressed in zero-centered transformed latent-rate coordinates. The perturbation applies only to latent shape; physical total number remains at its reference value. The rollout begins from \((\widetilde{\mathbf z}_{i,k},N^\star_{i,k})\) and advances recursively for \(h\) steps, so every predicted state supplies the rates for the next step. Validation, testing, and simulation set \(\boldsymbol\xi=\mathbf0\).

The same displacement is removed from the first latent-rate target, giving
\begin{equation}
\begin{aligned}
\widetilde{\mathbf u}^{\star}_{\dot z,i,k+j}
&:=\overline T_{\dot z}(\dot{\mathbf z}^\star_{i,k+j})
-\mathbb I[j=0]\,\boldsymbol\xi_{i,k},\\
u^\star_{\dot s,i,k+j}
&:=\overline T_{\dot s}(\dot s^\star_{i,k+j}),
\qquad j=0,\ldots,h-1.
\end{aligned}
\label{eq:reference-transformed-rates}
\end{equation}
This is the exact affine correction for the fitted \(\alpha_{\dot z}=\modelparam{transform.latent_rate_power}{1}\) transformation. The perturbation-and-correction construction is designed to expose the learned vector field to states near the encoded reference trajectory, including states that recursive integration can reach after small prediction errors. For the affine rate transformation used here, subtracting the correction from the first target supplies the rate displacement that cancels the sampled latent offset over one Euler step. Optimization therefore favors a restoring component in the local tendency near each training state before the rollout proceeds recursively from its own predictions. This is a structural bias on the neighborhood of the learned vector field rather than generic input-noise augmentation. Its intended role is to reduce amplification of small off-trajectory errors. The present experiments do not isolate that causal contribution, so recovery from perturbations and improved rollout stability remain design objectives; their independent effects await evaluation.

The corresponding predictions \(\widehat{\mathbf u}_{\dot z,i,k+j}\) and \(\widehat u_{\dot s,i,k+j}\) are evaluated at the recursively predicted state. The per-window rate losses compare the arithmetic means of the transformed rates:
\begin{equation}
\begin{aligned}
\ell_{\dot z,i,k}^{(h)}
&=\frac{1}{d_z}
\left\|
\frac1h\sum_{j=0}^{h-1}\widehat{\mathbf u}_{\dot z,i,k+j}
-\frac1h\sum_{j=0}^{h-1}\widetilde{\mathbf u}^\star_{\dot z,i,k+j}
\right\|_2^2,\\
\ell_{\dot s,i,k}^{(h)}
&=\left|
\frac1h\sum_{j=0}^{h-1}\widehat u_{\dot s,i,k+j}
-\frac1h\sum_{j=0}^{h-1}u^\star_{\dot s,i,k+j}
\right|^2.
\end{aligned}
\label{eq:window-mean-rate-loss}
\end{equation}
This objective constrains the net transformed tendency along the self-generated trajectory; it is distinct from an average of stepwise squared errors. Recursive rollout nevertheless couples early state errors to all later rate evaluations and to the endpoint FDT. In the selected model, direct state-rollout terms have \modelparam{training.direct_rollout_weight}{zero} weight. The three active loss families---latent rate, scale rate, and the complete endpoint FDT---have equal normalized total weight.

For horizon \(h\), let \(\mathcal B_{h,b}\), \(b=1,\ldots,B_h\), denote one optimizer batch of eligible rollout windows, and let \(\mathcal H=\{\modelparam{training.horizon_min_h}{1},\ldots,\modelparam{training.horizon_max_h}{12}\}\). The loss differentiated for that batch is
\begin{equation}
\begin{aligned}
\mathcal L_{\mathrm{batch}}^{(h,b)}
&=\omega_h\frac{1}{|\mathcal B_{h,b}|}
\sum_{(i,k)\in\mathcal B_{h,b}}
\left[
\modelparam{training.family_weight}{\frac13}\ell_{\dot z,i,k}^{(h)}
+\modelparam{training.family_weight}{\frac13}\ell_{\dot s,i,k}^{(h)}
+\modelparam{training.family_weight}{\frac13}\ell_{\mathrm{FDT},i,k}^{(h)}
\right],\\
\omega_h
&:=\frac{\sum_{r\in\mathcal H}B_r}{|\mathcal H|B_h}.
\end{aligned}
\label{eq:batch-rollout-objective}
\end{equation}
The factor \(\omega_h\) gives every rollout horizon equal total training weight despite different numbers of eligible batches. The final incomplete batch is retained and normalized by its actual number of windows. It receives the same horizon-specific batch factor as a full batch, so each window in that smaller batch carries correspondingly greater per-window weight. Because the selected run uses no gradient accumulation, Eq.~\eqref{eq:batch-rollout-objective} is backpropagated in one optimizer update for each batch. FDT is evaluated only at the horizon endpoint, whereas the two rate losses use the complete corrected-reference window means defined above. For monitoring, early stopping, and checkpoint selection only, training and validation terms are reduced over valid windows within each horizon and then averaged equally across horizons to form the reported epoch aggregate; that aggregate is not itself differentiated.

\hypertarget{data-training-and-evaluation}{%
\section{Data, Training, and Evaluation}\label{data-training-and-evaluation}}

\hypertarget{particle-resolved-coagulation-trajectories}{%
\subsection{Particle-resolved coagulation trajectories}\label{particle-resolved-coagulation-trajectories}}

The trajectories are drawn from the PartMC scenario library used to train the AeroMELD representation \cite{gasparik2020scenario,riemer2009partmc,zaveri2008mosaic,saleh2026aeromeld}. PartMC represents each aerosol population with a variable set of computational particles and resolves their individual sizes and compositions. The present experiments isolate binary coagulation at a fixed temperature of \modelparam{data.temperature_k}{293} K and pressure of \(\modelparam{data.pressure_pa}{100{,}000}\) Pa. Source states are available every \modelparam{data.source_interval_s}{600} s and are sampled hourly for latent-dynamics training and evaluation. Each of the \modelparam{data.total_trajectories}{20{,}000} trajectories spans \modelparam{data.duration_h}{48} h with \modelparam{data.states_per_trajectory}{49} hourly states; the dataset contains \modelparam{data.species_count}{15} chemical species, \modelparam{data.bin_count}{20} reporting size bins, and nominally \modelparam{data.nominal_particles}{1{,}000} computational particles per state.

The scenario library is built from diverse initial aerosol populations, with no reliance on repeated perturbations of a single trajectory. Training and held-out evaluation nevertheless share the same sampled scenario family and fixed environmental conditions. The evidence therefore applies to coagulation on this covered manifold; variable-environment, cross-process, and three-dimensional coupling are discussed as extensions in Section~\ref{discussion-and-conclusion}.

\begin{table}[H]
\centering
\caption{Scientific scope and size of the particle-resolved coagulation dataset.}
\label{tab:dataset-scope}
\begin{tabular}{@{}ll@{}}
\toprule
Quantity & Value \\
\midrule
Process & Binary coagulation \\
Temperature & \modelparam{data.temperature_k}{293} K \\
Pressure & \(\modelparam{data.pressure_pa}{100{,}000}\) Pa \\
Source-state interval & \modelparam{data.source_interval_s}{600} s \\
Latent-dynamics interval & \modelparam{training.integration_step_h}{1} h \\
Duration and states & \modelparam{data.duration_h}{48} h; \modelparam{data.states_per_trajectory}{49} states per trajectory \\
Trajectories and states & \modelparam{data.total_trajectories}{20{,}000}; \modelparam{data.total_states}{980{,}000} \\
Chemical species & \modelparam{data.species_count}{15} \\
Reporting size bins & \modelparam{data.bin_count}{20} \\
Nominal computational particles & \modelparam{data.nominal_particles}{1{,}000} per state \\
\bottomrule
\end{tabular}
\end{table}

\hypertarget{trajectory-level-partitioning}{%
\subsection{Trajectory-level partitioning}\label{trajectory-level-partitioning}}

Trajectories are partitioned before any data-dependent transformation is fitted, and every state from one physical trajectory remains in the same partition. The resulting sets contain \modelparam{split.train_trajectories}{16{,}000} training trajectories, \modelparam{split.validation_trajectories}{2{,}000} validation trajectories, and \modelparam{split.test_trajectories}{2{,}000} held-out test trajectories. The total-number reference scale, rate-transformation statistics, and all other fitted preprocessing quantities use the training partition only. This trajectory-level separation prevents temporally adjacent states from the same simulation from crossing the evaluation boundary.

\begin{table}[H]
\centering
\caption{Trajectory-level data partition used throughout model fitting and evaluation.}
\label{tab:data-partition}
\begin{tabular}{@{}lrr@{}}
\toprule
Partition & Trajectories & Hourly states \\
\midrule
Training & \modelparam{split.train_trajectories}{16{,}000} & \modelparam{split.train_states}{784{,}000} \\
Validation & \modelparam{split.validation_trajectories}{2{,}000} & \modelparam{split.validation_states}{98{,}000} \\
Held-out test & \modelparam{split.test_trajectories}{2{,}000} & \modelparam{split.test_states}{98{,}000} \\
\bottomrule
\end{tabular}
\end{table}

\hypertarget{frozen-aeromeld-encoder-and-decoder}{%
\subsection{Frozen AeroMELD Encoder and Decoder}\label{frozen-aeromeld-encoder-and-decoder}}

The Encoder--Decoder follows AeroMELD's scale--shape construction \cite{saleh2026aeromeld}. A particle network with \modelparam{encoder.hidden_layers}{two} \modelparam{encoder.hidden_width}{128}-unit ReLU hidden layers maps transformed per-species particle masses to learned features, which are averaged with normalized particle-number weights. The deterministic mean of the \modelparam{model.latent_dim}{nine}-dimensional latent shape is used for every reference state and rollout. The Decoder uses \modelparam{decoder.hidden_layers}{two} \modelparam{decoder.hidden_width}{128}-unit ReLU hidden layers and a \modelparam{decoder.output_dim}{550}-component output layer whose block-specific inverse transformations reconstruct the quantities listed in Section~\ref{aerosol-population-and-scale-shape-representation}. Encoder and Decoder weights and all fitted transformations are fixed before latent-dynamics training, so the evaluation isolates dynamics on one representation.

CCN activation follows the particle-resolved treatment used by Riemer et al.~\cite{riemer2010aging}, and the optical quantities follow a Mie-based mixing-state calculation \cite{yao2022}. Freezing follows the ice-nucleation-active-site-density framework reviewed by Hoose and M{\"o}hler \cite{hoose2012} and uses the dust and black-carbon parameterizations of Niemand et al.~and Schill et al., respectively, within their stated temperature ranges \cite{niemand2012inas,schill2020blackcarbon}; the treatment of internally mixed, multi-species particles follows Tang et al.~\cite{tang2026freezing}.

\paragraph{Sectional baseline.}
We compare AeroMELD with the PartMC sectional Jacobson two-moment
implementation initialized from exactly the same held-out PartMC populations,
matched by trajectory identity. The sectional state uses
\modelparam{baseline.sectional.bin_count}{20} logarithmically spaced wet-diameter
bins from \modelparam{baseline.sectional.diameter_min_um}{0.001} to
\modelparam{baseline.sectional.diameter_max_um}{10}~\(\mu\mathrm m\), with a
constant edge ratio of \modelparam{baseline.sectional.bin_ratio}{1.585}. Each
bin advances number concentration and the mass concentrations of
\modelparam{baseline.sectional.species_count}{15} aerosol species, giving
\modelparam{baseline.sectional.variables_per_bin}{16} variables per bin and
\modelparam{baseline.sectional.state_dim}{320} prognostic coordinates under the
state-count convention used in Figure~\ref{fig:efficiency_accuracy}. The
sectional and AeroMELD evaluations use the same coagulation environment, 48-h
horizon, hourly reporting times, and particle-resolved PartMC reference.
Sectional CCN, optical, and freezing diagnostics are computed by reconstructing
20 weighted sub-bin particles per bin that preserve bin number, wet volume,
and species fractions. Both reduced models then use the same diagnostic
calculations and symmetric error metric. The separate timing benchmark
advances PyPartMC and the sectional model with 60-s steps and disables model
output; the resources and timing scope are reported in
Section~\ref{computational-throughput}.

\hypertarget{optimization-and-model-selection}{%
\subsection{Optimization and model selection}\label{optimization-and-model-selection}}

The latent-dynamics network is optimized with AdamW using learning rate \(\modelparam{training.learning_rate}{10^{-4}}\), weight decay \(\modelparam{training.weight_decay}{10^{-4}}\), batch size \modelparam{training.batch_size}{64}, and gradient-norm clipping at \modelparam{training.gradient_clip}{1.0} \cite{loshchilov2019decoupled}. Validation is evaluated after each epoch, and \modelparam{training.checkpoint_selection}{the final training checkpoint} is used. Every training epoch assigns equal total weight to horizons from \modelparam{training.horizon_min_h}{1} through \modelparam{training.horizon_max_h}{12} h; validation and simulation use the same \modelparam{training.integration_step_h}{one}-hour forward-Euler step without training perturbations.

\hypertarget{evaluation-quantities-and-error-metric}{%
\subsection{Evaluation quantities and error metric}\label{evaluation-quantities-and-error-metric}}

We evaluate the latent shape \(\mathbf z\), physical total number \(N\), and all seven decoded blocks. For a predicted vector \(\widehat{\mathbf q}\) and reference vector \(\mathbf q^\star\), the symmetric relative error is
\begin{equation}
e(\widehat{\mathbf q},\mathbf q^\star)
=\frac{\|\widehat{\mathbf q}-\mathbf q^\star\|_2}
{\|\widehat{\mathbf q}\|_2+\|\mathbf q^\star\|_2+10^{-30}}.
\label{eq:symmetric-relative-error}
\end{equation}
The scalar definition for \(N\) replaces the norms by absolute values. For CCN activation, evaluation is restricted to supersaturations from \(0.1\%\) to \(0.6\%\). Optical coefficients are converted to \(\mathrm{Mm}^{-1}\), floored at 0.1, and transformed with \(\log_{10}\); frozen fractions use \(\log_{10}(q+10^{-30})\). Appendix~\ref{diagnostic-evaluation-domains} records the complete domains.

The rollout is initialized from the encoded reference state, so the \(t=0\) errors in \(\mathbf z\) and \(N\) are zero and are omitted from their distribution plots. Decoded physical quantities at \(t=0\) are retained; they measure the frozen representation's reconstruction error before the latent rollout introduces any additional discrepancy.

\hypertarget{results}{%
\section{Results}\label{results}}

\hypertarget{a-48-h-coagulation-rollout-in-latent-and-physical-space}{%
\subsection{A representative 48-h coagulation rollout in latent and physical space}\label{a-48-h-coagulation-rollout-in-latent-and-physical-space}}

Figure~\ref{fig:coag_rollout_example} presents the held-out trajectory at the
50th-percentile position in the evaluation-only composite MSE ranking used for
illustrative case selection. This score gives equal weight to the nine
displayed quantity groups and averages across the evaluated times; it is
distinct from both the training objective and the symmetric relative error in
Eq.~\eqref{eq:symmetric-relative-error}. The empirical error rank alone determined the case selection, placing this trajectory at typical overall performance under the selection metric.

The prediction tracks the monotonic decline in total number and the dominant
temporal variation of the latent coordinates. The decoded outputs capture the
shift toward larger particles and the accompanying changes in CCN activation,
optical properties, frozen fraction, and size-resolved composition, while
localized discrepancies remain in some latent coordinates and fine
bin-resolved structure. Corresponding examples at the 10th- and
90th-percentile positions are provided in
Appendix~\ref{illustrative-rollouts-across-error-distribution}.
Section~\ref{latent-rollout-accuracy-and-decoded-physical-errors} next
quantifies performance across the complete held-out population.

The sectional curves provide a conventional reduced-order reference for the
same initial population. Both reduced models reproduce the total-number decline
and the shift in binned number. At 48 h, AeroMELD is closer to PartMC for total
number, binned number, and frozen fraction. The sectional model is closer for
CCN activation and for both optical coefficients. The composition
panels show AeroMELD and PartMC only. This single trajectory illustrates how
the representations distribute their errors; the population-level comparison
in Figure~\ref{fig:error_distributions} determines how those patterns extend
across the held-out set.

\begin{figure}[H]
\centering
\includegraphics[width=\textwidth]{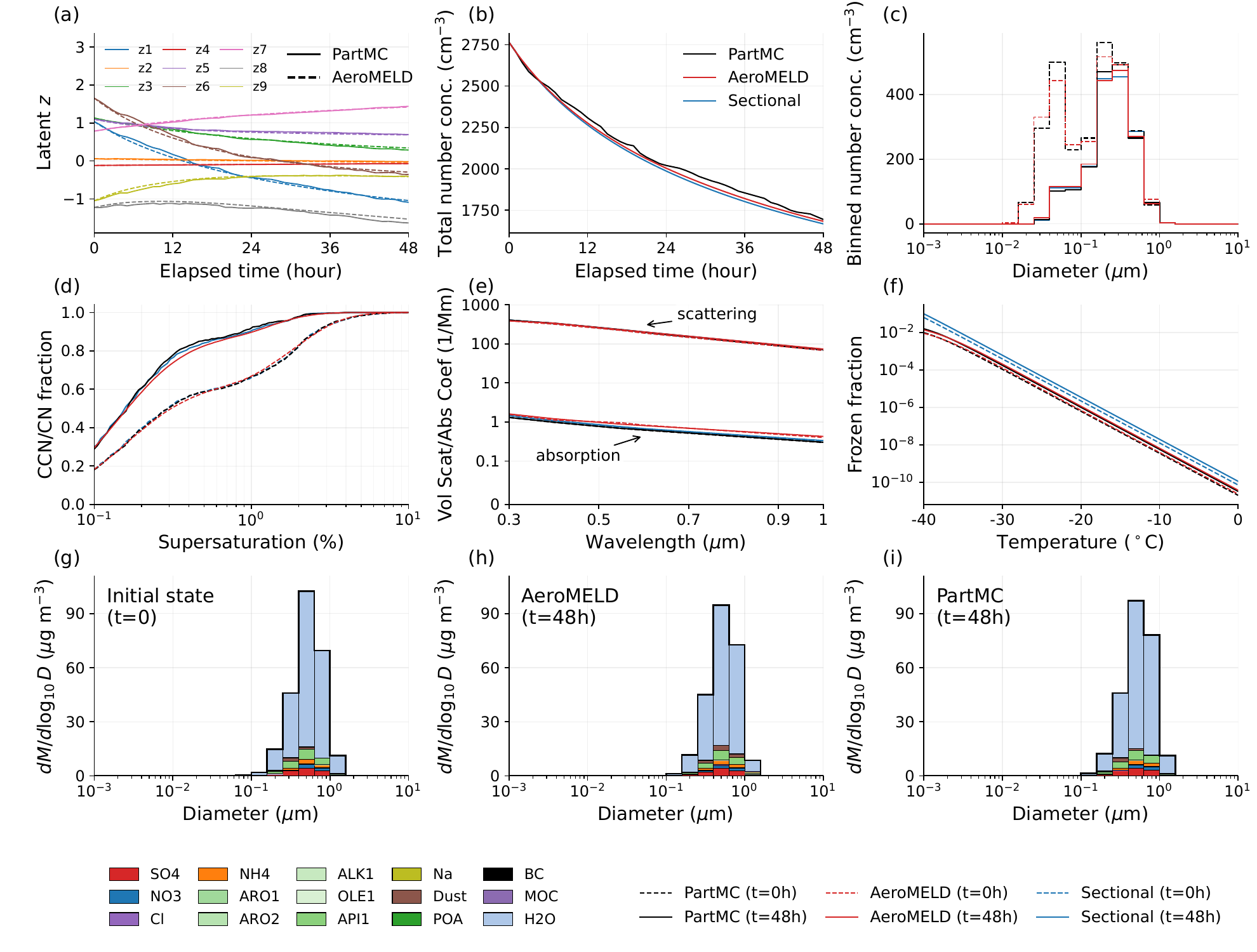}
\caption{Illustrative 48-h coagulation rollout in latent and decoded physical
space for the held-out case at the 50th-percentile position in the
evaluation-only composite MSE ranking among
\modelparam{split.test_trajectories}{2{,}000} trajectories. The ranking gives
equal weight to the nine displayed quantity groups and averages across the
evaluated times.
(a) \modelparam{model.latent_dim}{Nine} latent-shape coordinates; solid curves denote encoded PartMC
reference states and dashed curves denote the latent-dynamics rollout.
(b) Physical total number concentration, with the reference in black, the
prediction in red, and the sectional-model result in blue. (c--f) Binned number
concentration, CCN-to-CN ratio, aerosol optical properties, and frozen fraction.
Dashed and solid curves denote the initial and 48-h states, respectively.
Black, red, and blue denote reference, prediction, and sectional-model results.
Arrows distinguish scattering and absorption in
the optical panel. (g--i) Size-resolved composition at the initial reference
state, predicted 48-h state, and reference 48-h state. Black outlines show
total binned mass \(M_b\), and colors show normalized species fractions
\(P_{ab}\). Lower-error and upper-tail examples selected by the same ranking
are shown in Appendix~\ref{illustrative-rollouts-across-error-distribution};
aggregate held-out errors are reported in
Figures~\ref{fig:error_distributions} and~\ref{fig:error_vs_time}.}
\label{fig:coag_rollout_example}
\end{figure}

\hypertarget{latent-rollout-accuracy-and-decoded-physical-errors}{%
\subsection{Latent-rollout accuracy and decoded physical errors}\label{latent-rollout-accuracy-and-decoded-physical-errors}}

The direct state-space metrics show that the learned vector field advances the compact coagulation state accurately across the held-out population. Over all post-initial trajectory--time pairs, the median symmetric error in latent shape is \(\modelparam{results.latent.pooled_median_pct}{3.5\%}\), with a 95th percentile of \(\modelparam{results.latent.pooled_p95_pct}{8.9\%}\), while the median total-number error is \(\modelparam{results.number.pooled_median_pct}{0.52\%}\). The calculation uses state-space outputs from all \modelparam{split.test_trajectories}{2{,}000} recursive 48-h rollouts independently of the decoded quantities. The bounded time-dependent distributions in Figure~\ref{fig:error_vs_time} provide further support for the central proof of concept: within the sampled scenario manifold, the \modelparam{model.latent_dim}{nine}-dimensional latent shape carries a learnable and stable nonlinear coagulation trajectory.

Figures~\ref{fig:error_distributions} and~\ref{fig:error_vs_time} compare decoded predictions directly with the raw particle-resolved reference outputs, without passing those references through the frozen Encoder--Decoder. Their physical-space discrepancies therefore combine the representation's reconstruction error with the additional discrepancy introduced by latent rollout. The decomposition in Appendix~\ref{ed-reconstruction-and-ld-discrepancy} shows that the pooled transformed-output reconstruction RMSE is \modelparam{results.ed_ld.pooled_ratio}{4.9} times the additional latent-dynamics discrepancy and is larger in all \modelparam{decoder.block_count}{seven} decoded blocks. At 12, 24, and 48 h, the reconstruction-to-rollout discrepancy ratio is also at least three. The frozen representation therefore sets the dominant pooled decoded-error floor at these preregistered lead times, and the rollout adds a smaller drift. The comparison concerns residual magnitudes; additive percentage attribution of total squared error would also require the cross term.

The decoded blocks retain different error scales and tail widths. Size-resolved composition, binned mass, and absorption have broader distributions than total number, CCN activation, scattering, or frozen fraction, but most median and interquartile ranges remain bounded over the evaluated period. These block-to-block differences reflect representation sensitivity and dynamics; scientific importance is a separate consideration. The broader tails define the present accuracy boundary without changing the direct evidence that the latent trajectory remains stable over the evaluated 48 h.

\begin{figure}[H]
\centering
\includegraphics[width=\textwidth]{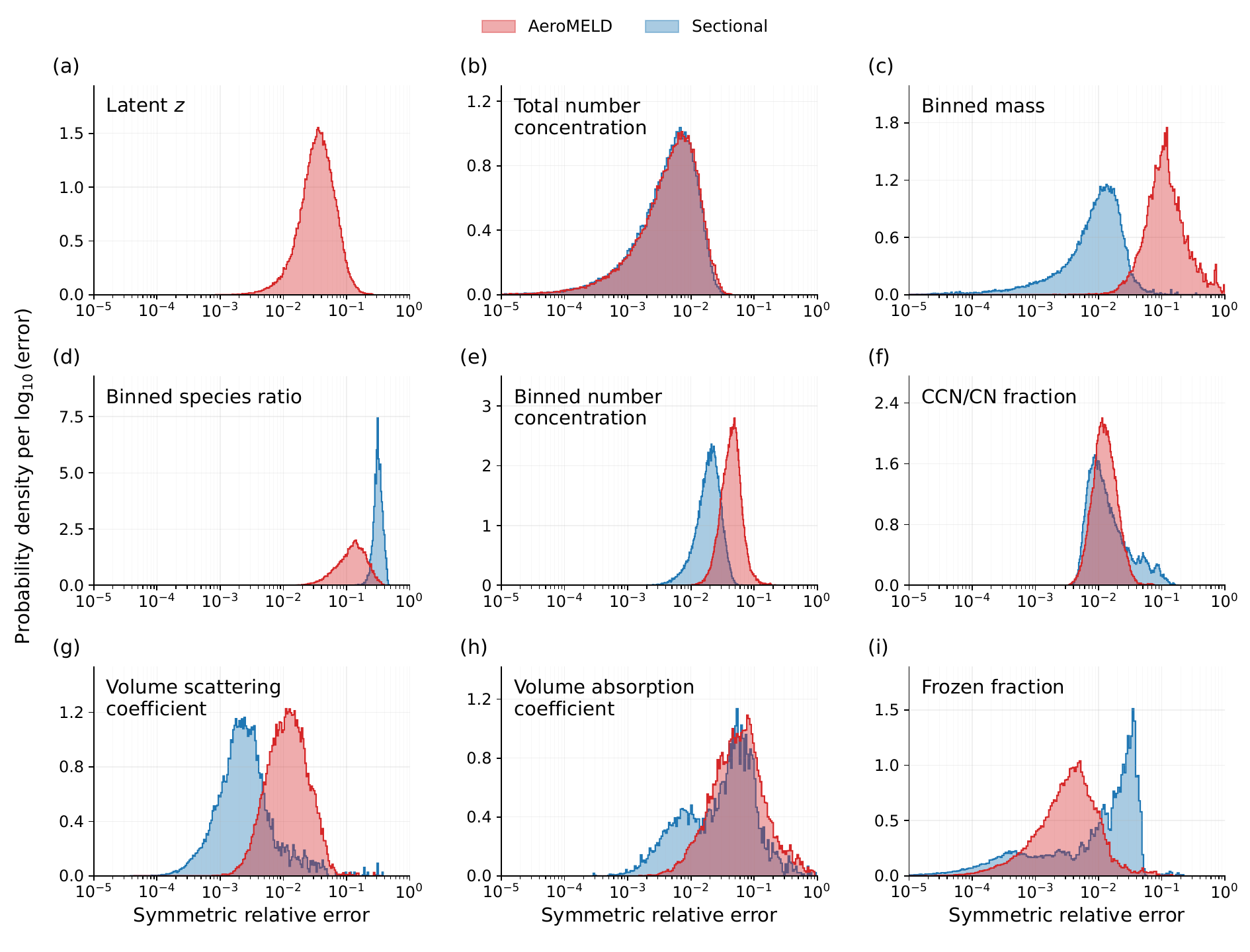}
\caption{Distribution of symmetric rollout errors across
\modelparam{split.test_trajectories}{2{,}000} held-out trajectories. The
decoded panels include all \modelparam{data.states_per_trajectory}{49} hourly
states from 0 to 48 h. The latent-shape and total-number panels include the 48
post-initial states because initialization fixes their \(t=0\) errors to zero.
Every included trajectory--time pair has equal weight. Panel (a) shows the
AeroMELD latent-shape error. Panels (b--i) compare AeroMELD (red) and the
sectional model (blue) against the same particle-resolved PartMC reference for
total number \(N\), total binned mass \(M_b\), normalized size-resolved
composition \(P_{ab}\), binned number \(N_b\), CCN-to-CN ratio, volume
scattering coefficient, volume absorption coefficient, and frozen fraction.
Table~\ref{tab:pooled-symmetric-error-summary} reports the pooled mean and
2.5th--97.5th empirical percentile interval for the same quantities. Decoded
\(t=0\) errors are retained and include frozen-representation reconstruction
error. Section~\ref{evaluation-quantities-and-error-metric} defines the metric
and diagnostic transformations.}
\label{fig:error_distributions}
\end{figure}

\begin{table}[H]
\centering
\caption{Pooled symmetric relative errors for the nine evaluation quantities
in Figure~\ref{fig:error_distributions}. Values are percentages and are
reported as mean [2.5th percentile, 97.5th percentile] over the same
trajectory--time pairs used in that figure. Brackets denote empirical
distribution percentiles. The latent-shape entry applies only to AeroMELD.
Boldface marks the lower pooled mean when the two models show a clear
difference; the similar total-number means are left unbolded.}
\label{tab:pooled-symmetric-error-summary}
\footnotesize
\begin{tabularx}{\textwidth}{@{}Xcc@{}}
\toprule
Quantity & \multicolumn{2}{c}{Symmetric relative error of} \\
 & AeroMELD & Sectional model \\
\midrule
Latent shape \(\mathbf z\) &
\modelparam{results.latent.pooled_mean_pct}{4.1\%}
[\modelparam{results.latent.pooled_p2_5_pct}{0.786\%},
 \modelparam{results.latent.pooled_p97_5_pct}{10.6\%}] &
N/A \\
Total number concentration \(N\) &
\textbf{\modelparam{results.number.pooled_mean_pct}{0.656\%}}
[\modelparam{results.number.pooled_p2_5_pct}{0.0231\%},
 \modelparam{results.number.pooled_p97_5_pct}{2.06\%}] &
\textbf{\modelparam{results.sectional.total_number.pooled_mean_pct}{0.614\%}}
[\modelparam{results.sectional.total_number.pooled_p2_5_pct}{0.0224\%},
 \modelparam{results.sectional.total_number.pooled_p97_5_pct}{1.87\%}] \\
Total binned mass \(M_b\) &
\modelparam{results.decoded.mass.pooled_mean_pct}{15\%}
[\modelparam{results.decoded.mass.pooled_p2_5_pct}{2.76\%},
 \modelparam{results.decoded.mass.pooled_p97_5_pct}{59.7\%}] &
\textbf{\modelparam{results.sectional.mass.pooled_mean_pct}{1.2\%}}
[\modelparam{results.sectional.mass.pooled_p2_5_pct}{0.00372\%},
 \modelparam{results.sectional.mass.pooled_p97_5_pct}{3.68\%}] \\
Normalized size-resolved composition \(P_{ab}\) &
\textbf{\modelparam{results.decoded.composition.pooled_mean_pct}{13.5\%}}
[\modelparam{results.decoded.composition.pooled_p2_5_pct}{3.61\%},
 \modelparam{results.decoded.composition.pooled_p97_5_pct}{28.6\%}] &
\modelparam{results.sectional.composition.pooled_mean_pct}{31.2\%}
[\modelparam{results.sectional.composition.pooled_p2_5_pct}{19.1\%},
 \modelparam{results.sectional.composition.pooled_p97_5_pct}{41.9\%}] \\
Binned number concentration \(N_b\) &
\modelparam{results.decoded.binned_number.pooled_mean_pct}{4.62\%}
[\modelparam{results.decoded.binned_number.pooled_p2_5_pct}{2.03\%},
 \modelparam{results.decoded.binned_number.pooled_p97_5_pct}{8.85\%}] &
\textbf{\modelparam{results.sectional.binned_number.pooled_mean_pct}{1.99\%}}
[\modelparam{results.sectional.binned_number.pooled_p2_5_pct}{0.411\%},
 \modelparam{results.sectional.binned_number.pooled_p97_5_pct}{3.86\%}] \\
CCN-to-CN ratio &
\textbf{\modelparam{results.decoded.ccn.pooled_mean_pct}{1.38\%}}
[\modelparam{results.decoded.ccn.pooled_p2_5_pct}{0.555\%},
 \modelparam{results.decoded.ccn.pooled_p97_5_pct}{2.95\%}] &
\modelparam{results.sectional.ccn.pooled_mean_pct}{1.98\%}
[\modelparam{results.sectional.ccn.pooled_p2_5_pct}{0.532\%},
 \modelparam{results.sectional.ccn.pooled_p97_5_pct}{8.64\%}] \\
Volume scattering coefficient &
\modelparam{results.decoded.scattering.pooled_mean_pct}{1.52\%}
[\modelparam{results.decoded.scattering.pooled_p2_5_pct}{0.277\%},
 \modelparam{results.decoded.scattering.pooled_p97_5_pct}{4.49\%}] &
\textbf{\modelparam{results.sectional.scattering.pooled_mean_pct}{0.547\%}}
[\modelparam{results.sectional.scattering.pooled_p2_5_pct}{0.0374\%},
 \modelparam{results.sectional.scattering.pooled_p97_5_pct}{2.81\%}] \\
Volume absorption coefficient &
\modelparam{results.decoded.absorption.pooled_mean_pct}{8.76\%}
[\modelparam{results.decoded.absorption.pooled_p2_5_pct}{0.81\%},
 \modelparam{results.decoded.absorption.pooled_p97_5_pct}{39.9\%}] &
\textbf{\modelparam{results.sectional.absorption.pooled_mean_pct}{5.49\%}}
[\modelparam{results.sectional.absorption.pooled_p2_5_pct}{0.203\%},
 \modelparam{results.sectional.absorption.pooled_p97_5_pct}{24.3\%}] \\
Frozen fraction &
\textbf{\modelparam{results.decoded.frozen.pooled_mean_pct}{0.56\%}}
[\modelparam{results.decoded.frozen.pooled_p2_5_pct}{0.0244\%},
 \modelparam{results.decoded.frozen.pooled_p97_5_pct}{2.43\%}] &
\modelparam{results.sectional.frozen.pooled_mean_pct}{1.72\%}
[\modelparam{results.sectional.frozen.pooled_p2_5_pct}{0.00753\%},
 \modelparam{results.sectional.frozen.pooled_p97_5_pct}{4.6\%}] \\
\bottomrule
\end{tabularx}
\end{table}

Figure~\ref{fig:error_distributions} and
Table~\ref{tab:pooled-symmetric-error-summary} quantify the accuracy trade-off
between two reduced aerosol states. Total-number errors are comparable. The
sectional baseline has lower errors for total binned mass, binned number,
scattering, and absorption. AeroMELD has substantially lower composition and
frozen-fraction errors. The CCN distributions have similar medians. AeroMELD
has a lower pooled mean and a much narrower CCN upper tail.

The error patterns reflect both the representations and the model
formulations. AeroMELD learns from the joint particle size--composition
distribution, and its frozen Decoder retains the selected diagnostics. The
sectional state carries the mixing-state detail encoded by its bin variables.
The models also differ in parameterization and numerical formulation. The
benchmark compares complete model configurations, including their
representations, parameterizations, and numerical formulations. Component-level
attribution lies outside the benchmark.

Each model combines distinct error sources. AeroMELD errors combine frozen
representation residual and latent-rollout discrepancy. Sectional errors
combine initialization and binning, diagnostic closure, and numerical
evolution errors. Both models are evaluated against the same
particle-resolved outputs with the same paper-facing metric, so
Figure~\ref{fig:error_distributions} and
Table~\ref{tab:pooled-symmetric-error-summary} compare end-to-end diagnostic
performance across their complete error pathways.

\begin{figure}[H]
\centering
\includegraphics[width=\textwidth]{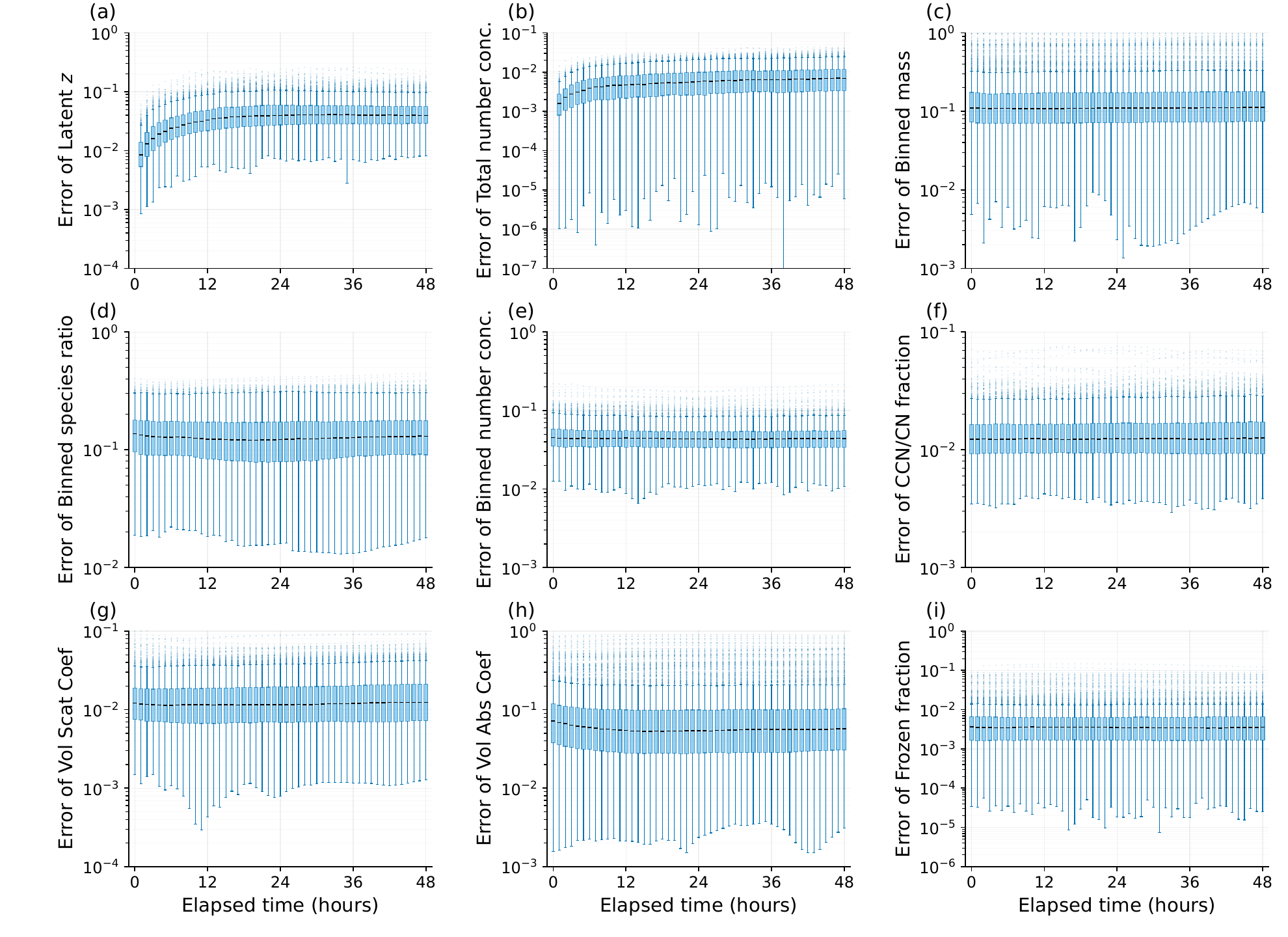}
\caption{Time evolution of symmetric rollout-error distributions across \modelparam{split.test_trajectories}{2{,}000} held-out trajectories. At each hourly state, boxes show the median and interquartile range, whiskers extend to 1.5 times the interquartile range, and points denote outliers for the same nine quantities and panel order as Figure~\ref{fig:error_distributions}. Vertical axes are logarithmic. The \(t=0\) values for \(\mathbf z\) and \(N\) are omitted because they are fixed by initialization; decoded \(t=0\) errors include frozen-representation reconstruction error.}
\label{fig:error_vs_time}
\end{figure}

Because binary coagulation conserves total aerosol mass, decoded mass
retention provides an additional budget diagnostic that is distinct from the
latent-shape and total-number errors. Panel (a) of
Figure~\ref{fig:total_mass_retention} normalizes each held-out rollout by its
own decoded initial total mass, thereby isolating temporal drift from any
frozen-representation offset at initialization. Relative to the ideal
unit-retention baseline, the median symmetric errors are \modelparam{results.mass.retention_24_median_pct}{1.7\%} at 24 h and
\modelparam{results.mass.retention_48_median_pct}{2.2\%} at 48 h; at 48 h the interquartile range is \modelparam{results.mass.retention_48_q25_pct}{0.95\%}--\modelparam{results.mass.retention_48_q75_pct}{4.11\%} and the 95th
percentile is \modelparam{results.mass.retention_48_p95_pct}{11.6\%}. Panel (b) compares decoded total mass directly with the PartMC reference and therefore includes both the static representation offset and rollout drift. The central retention distribution indicates modest mass drift for the present proof-of-concept rollouts. In the upper tail, the unconstrained model lacks strict mass closure for some trajectories.
Appendix~\ref{decoded-total-mass-retention} examines cases selected from the
lower and upper tails of the self-normalized decoded-mass trajectory extrema.
Their CCN, optical, and frozen-fraction errors remain much smaller than their
scalar mass discrepancies at 48 h. In these selected trajectories, large scalar
mass discrepancies coexist with much smaller diagnostic errors. Strict mass
conservation remains an unresolved model constraint.

The examples in Figures~\ref{fig:ccn_mixing_state_examples} and
\ref{fig:frozen_fraction_mixing_state_examples} explain why AeroMELD can be
closer to PartMC than the Sectional model. AeroMELD can retain particle-level
mixing-state information in its learned representation, whereas the Sectional
model assigns a mean composition to each size bin. Each figure compares a
complex mixing-state case, where this difference matters, with a simple case,
where both reduced models perform similarly.

For CCN activation (Figure~\ref{fig:ccn_mixing_state_examples}), PartMC
particles in the red-framed case span a wide range of critical
supersaturations at similar diameters because their compositions differ. The
Sectional model replaces this variability with the bin-mean composition,
producing a narrow diameter--critical-supersaturation relation and
overestimating the activated fraction. AeroMELD remains close to PartMC. In
the blue-framed case, the PartMC and Sectional
diameter--critical-supersaturation distributions are already similar, and both
reduced models reproduce the PartMC CCN spectrum.

\begin{figure}[H]
\centering
\includegraphics[width=\textwidth]{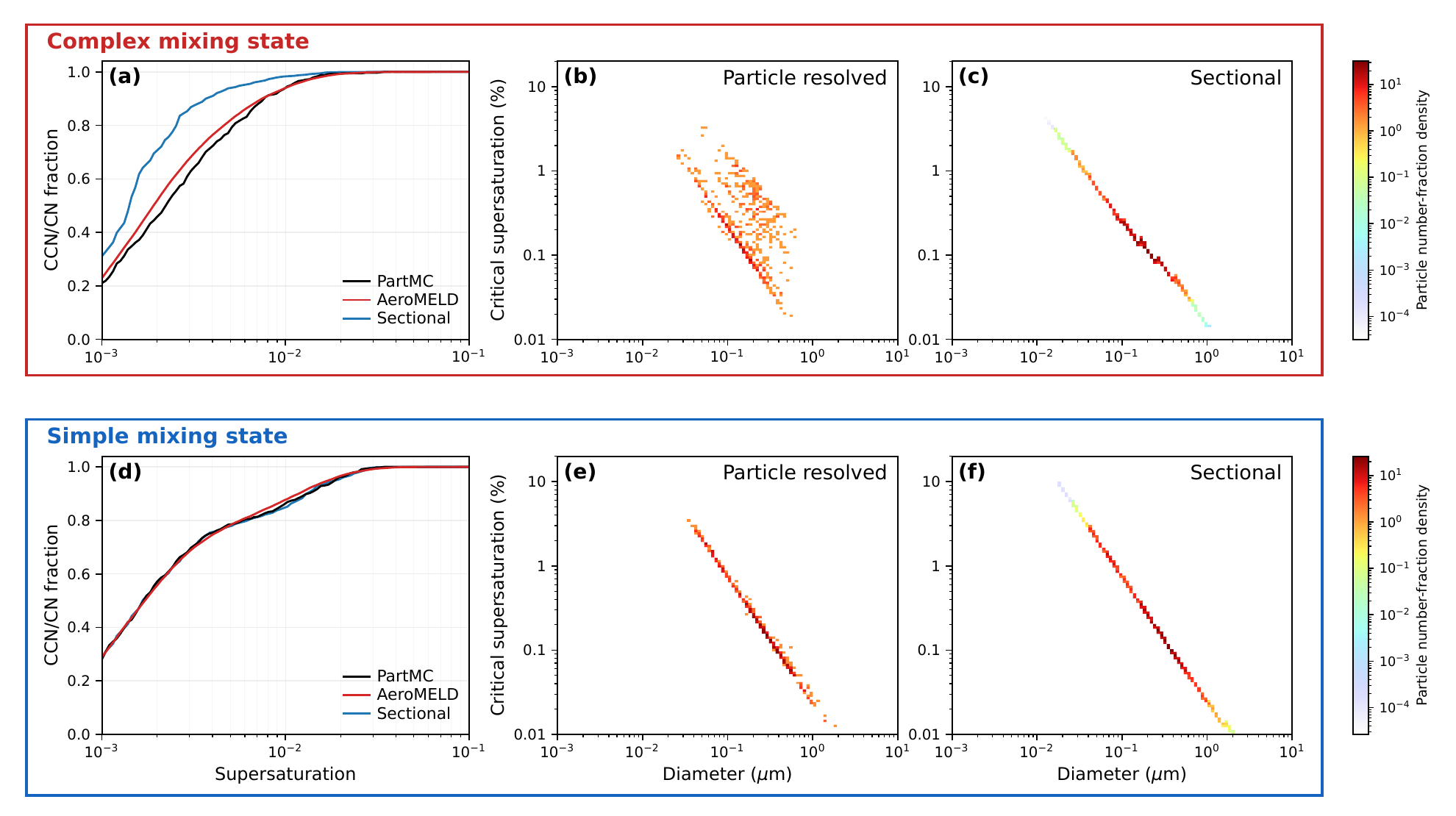}
\caption{CCN spectra and diameter--critical-supersaturation distributions at
48 h. The red and blue frames mark complex and simple mixing-state cases,
respectively. (a,d) CCN-to-CN fraction as a function of supersaturation for
PartMC (black), AeroMELD (red), and the Sectional model (blue). (b,e) PartMC
particle number-fraction densities over wet diameter and critical
supersaturation. (c,f) Corresponding Sectional distributions. Colors show
particle number-fraction density in the plotted coordinates.}
\label{fig:ccn_mixing_state_examples}
\end{figure}

The INP comparison in
Figure~\ref{fig:frozen_fraction_mixing_state_examples} shows the same loss of
mixing-state information more directly. In the red-framed case, most PartMC
particles contain no OIN, while a small subset has high OIN mass fractions.
The Sectional model spreads the bin-mean OIN content across all representatives
in each size bin, moving the population toward the internal-mixing limit.
Within a fixed size bin and for a fixed bin-level species surface-area
composition, the internally mixed state gives the upper frozen-fraction bound,
whereas the externally mixed state gives the lower bound
\cite{tang2026freezing}. The Sectional frozen fraction is consequently
substantially higher than the PartMC result near the cold end of the spectrum,
while AeroMELD nearly overlaps PartMC. In the blue-framed case, the PartMC and
Sectional BC distributions have similar size-dependent patterns, and all three
frozen-fraction spectra agree closely.

\begin{figure}[H]
\centering
\includegraphics[width=\textwidth]{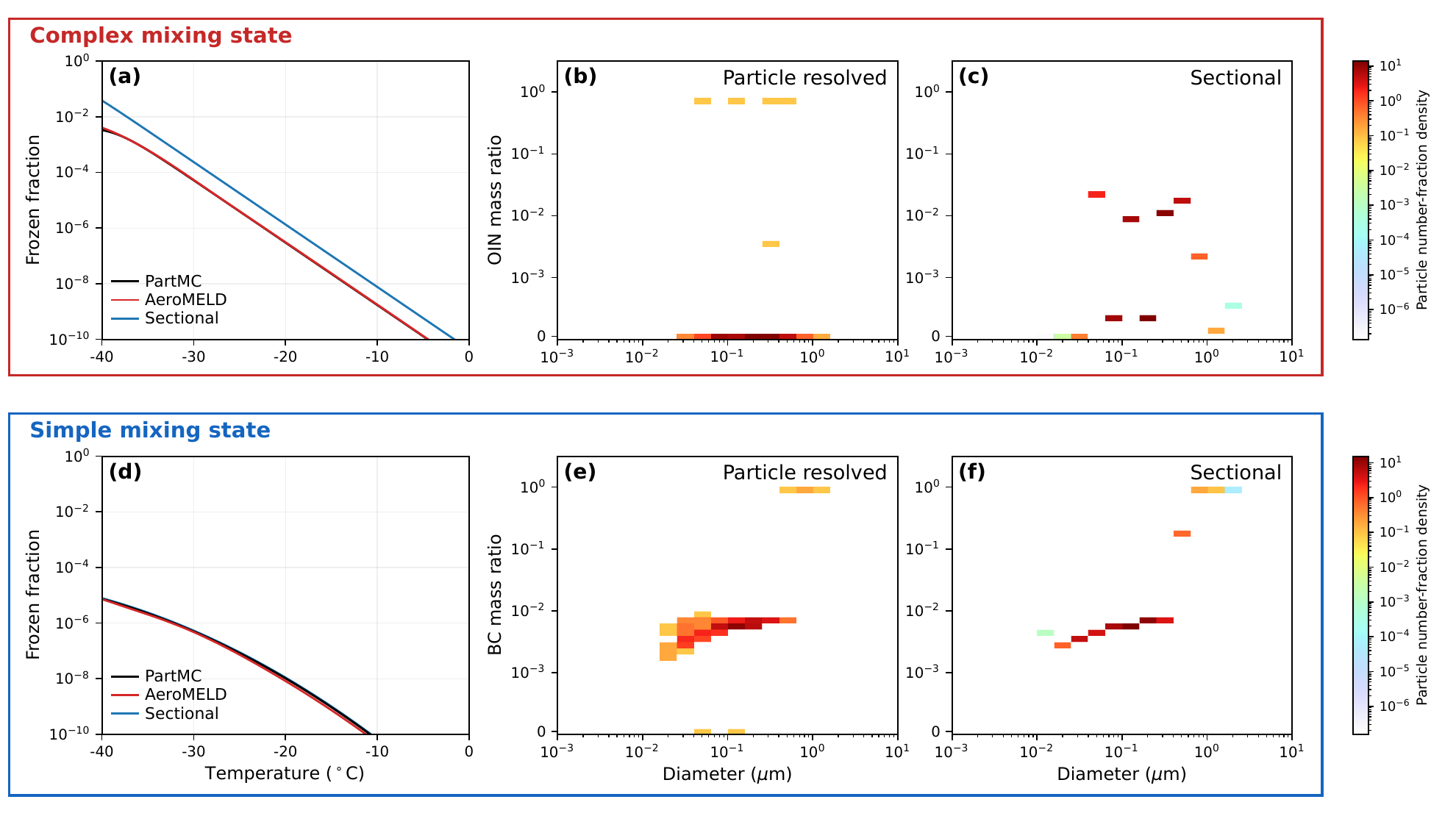}
\caption{Frozen-fraction spectra and insoluble-species mixing-state
distributions at 48 h. The red and blue frames mark complex and simple
mixing-state cases, respectively. (a,d) Frozen fraction as a function of
temperature for PartMC (black), AeroMELD (red), and the Sectional model (blue).
(b,c) PartMC and Sectional particle number-fraction densities over wet diameter
and OIN mass ratio for the complex case. (e,f) Corresponding distributions over
wet diameter and BC mass ratio for the simple case. OIN is shown in (b,c)
because it dominates freezing under the adopted parameterizations, whereas BC
is shown in (e,f) because OIN is absent and BC controls freezing in the simple
case. Colors show particle number-fraction density in the plotted coordinates.}
\label{fig:frozen_fraction_mixing_state_examples}
\end{figure}

\subsection{Matched integration benchmark}\label{computational-throughput}

The three models were timed on the same
\modelparam{split.test_trajectories}{2{,}000} initial conditions for 48 h of
Brownian coagulation with a 60-s step (2,880 steps) and model output disabled.
Mean integration time per case was 2.1866 s for PyPartMC and 0.7669 s for the
sectional model, each on one CPU core, and 0.4269 ms for AeroMELD, amortized
over a batch of 2,000 on one NVIDIA L40S. The AeroMELD timing excludes loading,
encoding, data transfer, decoding, validation, and file output. It measures
batched integration throughput rather than single-case latency, and the
CPU--GPU comparison is not an equal-hardware benchmark.

\hypertarget{integration-step-robustness}{%
\subsection{Integration-step robustness}\label{integration-step-robustness}}

The selected model was trained and normally evaluated with a \modelparam{training.integration_step_h}{1} h
forward-Euler step. Using this first-order, single-stage scheme provides a
deliberately conservative test of the learned continuous-time dynamics: each
step requires only one evaluation of the neural right-hand side.
Figure~\ref{fig:dt_stability} repeats the 48-h latent-shape and total-number
evaluations over fixed Euler steps from 1 s to 8 h and adds an 8 h RK4
comparison without retraining the model.

For the illustrative held-out case, the latent paths in panel (a) remain
closely aligned across the 10 min, 1 h, and 4 h Euler rollouts, with larger but
still finite differences at 8 h. The corresponding total-number curves in
panel (c) nearly overlap for 10 min and 1 h, remain close at 4 h, and separate
visibly at 8 h. Across all \modelparam{split.test_trajectories}{2{,}000} held-out trajectories, the median and
interquartile range of the 48-h latent-shape error are nearly unchanged from
1 s through 4 h and increase only modestly for 8 h Euler. Total-number errors
remain low through 2 h, increase at 4 h, and broaden more clearly for 8 h
Euler. All \modelparam{split.test_trajectories}{2{,}000} trajectories remain finite through 48 h for every tested Euler step and for 8 h RK4; the 8 h Euler integration remains numerically stable despite its loss of accuracy.

Changing only the numerical solver at \(\Delta t=8\) h lowers the median
latent-shape error from \modelparam{results.dt.euler8.latent_median_pct}{5.67\%} with Euler to \modelparam{results.dt.rk4_8.latent_median_pct}{3.98\%} with RK4 and lowers the
median total-number error from \modelparam{results.dt.euler8.number_median_pct}{2.30\%} to \modelparam{results.dt.rk4_8.number_median_pct}{0.73\%}. The improvement is most pronounced for total number, showing that a higher-order integrator can reduce the coarse-step degradation while using the same learned trajectory. The results show a numerically accommodating latent ODE. The inexpensive Euler update remains robust across a broad seconds-to-hours range, and the same learned vector field can be integrated with a higher-order solver when an unusually coarse step demands greater accuracy.
At the 1 h step used here, the experiment shows numerically stable integration without a finer internal substep for this coagulation operator.

\begin{figure}[H]
\centering
\includegraphics[width=\textwidth]{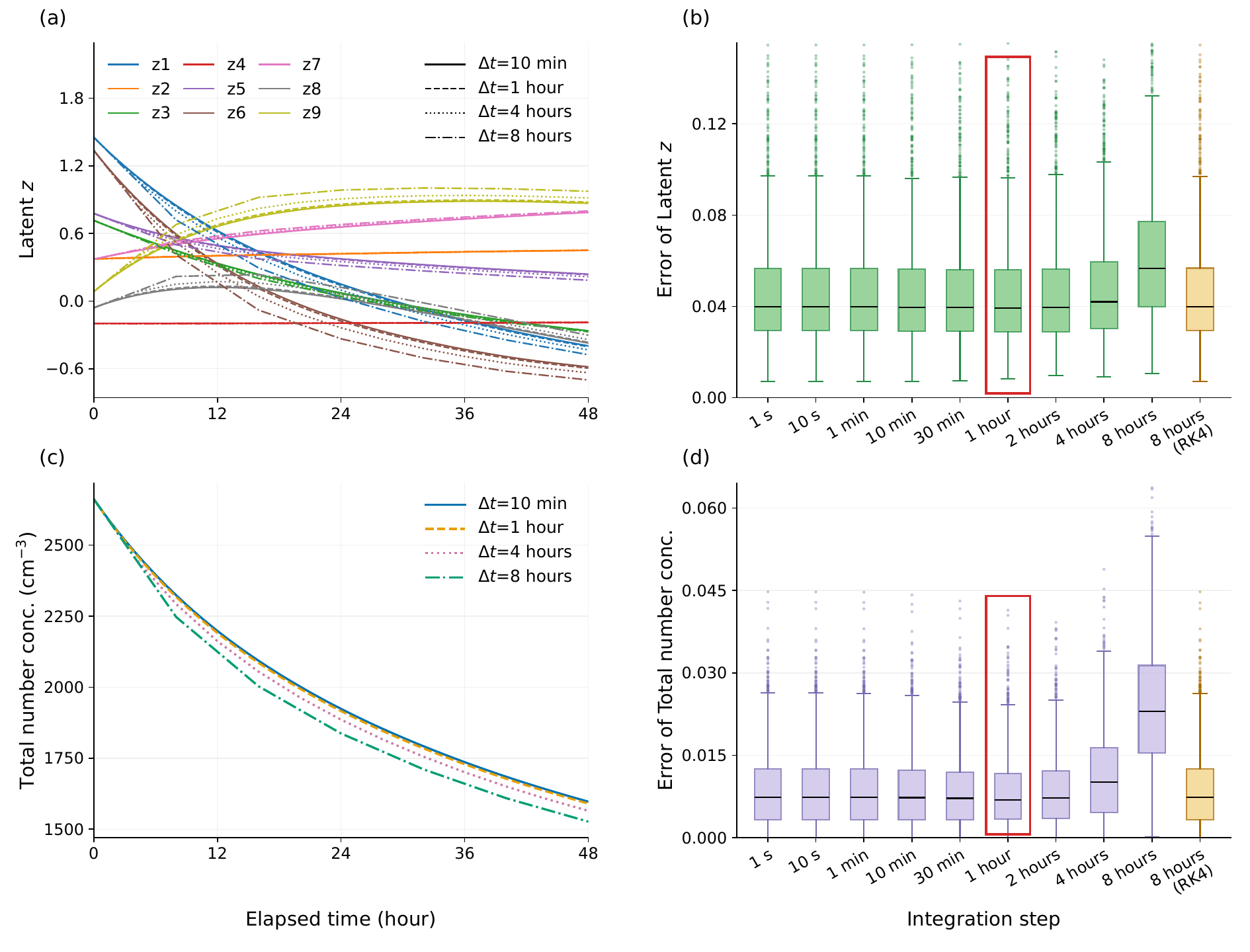}
\caption{Integration-step robustness and solver flexibility of the neural-ODE
rollout.
(a) AeroMELD latent-shape rollouts for one illustrative held-out case. Colors
identify \(z_1\)--\(z_{\modelparam{model.latent_last_index}{9}}\), and line styles identify fixed forward-Euler steps
of \(\Delta t=10\) min, 1 h, 4 h, and 8 h. (b) Distributions across \modelparam{split.test_trajectories}{2{,}000}
held-out trajectories of the 48-h latent-shape symmetric relative error
against the encoded PartMC reference. The first nine boxes use forward Euler
with steps from 1 s to 8 h; the final box uses classical fourth-order
Runge--Kutta (RK4) with \(\Delta t=8\) h. (c) Total number concentration for
the same illustrative case and the same four Euler steps as in panel (a). (d)
Corresponding distributions of the 48-h total-number symmetric relative error
against the PartMC reference, again comparing the Euler step-size sweep with
8 h RK4. Red outlines in panels (b) and (d) mark \(\Delta t=\modelparam{training.integration_step_h}{1}\) h, the only
integration step used during training. Boxes show medians and interquartile
ranges; whiskers extend to \(1.5\) times the interquartile range, and points
denote outliers. The RK4 comparison uses the same learned vector field and initial states, with the numerical solver as the sole change.}
\label{fig:dt_stability}
\end{figure}

\hypertarget{effect-of-rollout-aware-training}{%
\subsection{Effect of rollout-aware training}\label{effect-of-rollout-aware-training}}

An otherwise matched single-step model provides a direct test of rollout-aware
training. Figure~\ref{fig:single_multi_step} compares the two models over the
same \modelparam{split.test_trajectories}{2{,}000} held-out 48-h trajectories using both illustrative rollouts and
population-level error summaries. The frozen representation, latent-dynamics
architecture, test cases, and one-hour forward-Euler evaluation are shared;
the comparison changes the confirmed single-step versus multistep training
configuration.

\begin{figure}[H]
\centering
\includegraphics[width=\textwidth]{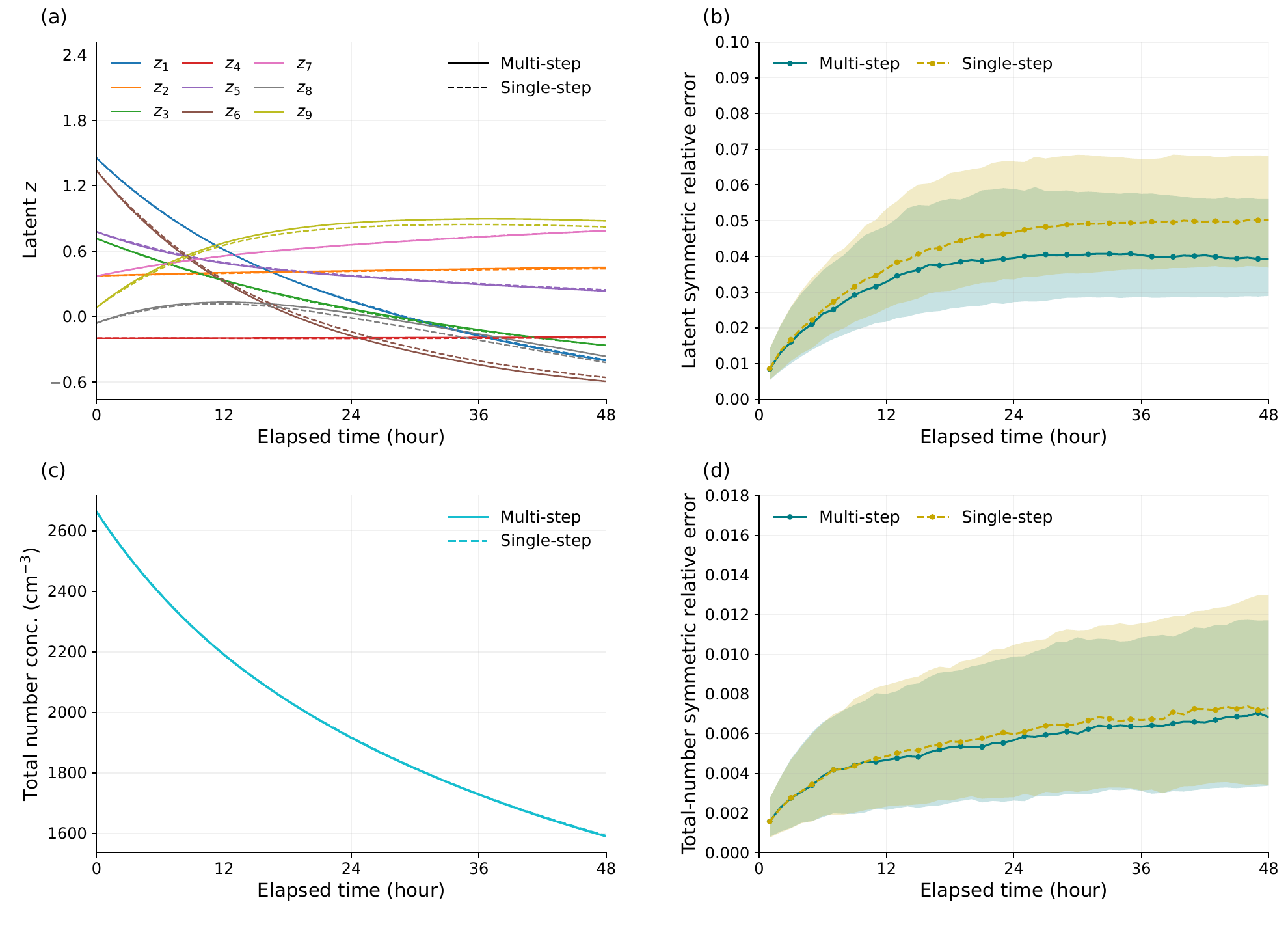}
\caption{Matched multistep--single-step comparison over 48-h held-out
rollouts. (a) \modelparam{model.latent_dim}{Nine} latent-shape coordinates for one illustrative held-out
trajectory. Colors identify \(z_1\)--\(z_{\modelparam{model.latent_last_index}{9}}\), while solid and dashed curves
denote the multistep and single-step models, respectively. (b) Time evolution
across all \modelparam{split.test_trajectories}{2{,}000} held-out trajectories of the mean, median, 25th percentile, and
75th percentile of the latent-shape symmetric relative error at each
post-initial hourly state. (c) Physical total number concentration for the
same illustrative trajectory as in panel (a), with the multistep and
single-step models again shown by solid and dashed curves. (d) Corresponding
mean, median, 25th-percentile, and 75th-percentile total-number symmetric
relative errors across the \modelparam{split.test_trajectories}{2{,}000} held-out trajectories. Both models start from
the same reference latent shape and total number; their zero initial errors
are omitted from panels (b) and (d).}
\label{fig:single_multi_step}
\end{figure}

The latent-error summaries separate progressively with lead time. At 48 h,
the mean latent symmetric relative error is \modelparam{results.matched.multistep.latent_48_mean_pct}{4.54\%} for the multistep model and
\modelparam{results.matched.single_step.latent_48_mean_pct}{5.58\%} for the single-step model, an \modelparam{results.matched.latent_48_mean_reduction_pct}{18.6\%} reduction. The corresponding
medians are \modelparam{results.matched.multistep.latent_48_median_pct}{3.92\%} and \modelparam{results.matched.single_step.latent_48_median_pct}{5.03\%}, a \modelparam{results.matched.latent_48_median_reduction_pct}{22.0\%} reduction. The 25th- and 75th-percentile curves show the same ordering at medium and long lead times, confirming that the improvement extends across the central error distribution. The total-number results differ much less: the two trajectories nearly overlap, and the 48-h error summaries remain similar, with mean errors of \modelparam{results.matched.multistep.number_48_mean_pct}{0.84\%} and \modelparam{results.matched.single_step.number_48_mean_pct}{0.90\%} and median errors of \modelparam{results.matched.multistep.number_48_median_pct}{0.68\%} and \modelparam{results.matched.single_step.number_48_median_pct}{0.73\%} for the multistep and single-step models, respectively. These small differences are not consistently ordered across
lead times and summary statistics. The measurable benefit of rollout-aware
training is therefore concentrated in the nine-dimensional latent-shape
evolution; both training strategies learn the scalar total-number decline
comparably well.

This per-lead-time symmetric-error comparison is complementary to the
per-trajectory latent RMSE analysis in
Appendix~\ref{multistep-single-step-tables}, which summarizes the first 24 h
across all nine latent coordinates and lead times. Under that metric, the
24-h mean latent RMSE decreases from \modelparam{results.matched.single_step.lead_24_mean_trajectory_rmse}{0.0734} to \modelparam{results.matched.multistep.lead_24_mean_trajectory_rmse}{0.0625}, a \modelparam{results.matched.lead_24_mean_trajectory_reduction_pct}{14.9\%} reduction, and
the largest relative gains occur in the upper tail. Both models remain finite over the evaluated 48-h rollouts. The evidence supports improved long-lead latent accuracy and tail robustness; stable coagulation rollouts are achieved by both training strategies.

The growing separation in latent error with lead time is consistent with the
purpose of rollout exposure: during multistep training, prediction errors can
feed back through the learned vector field before the objective is evaluated
at later states. The much smaller contrast in total-number error locates this advantage primarily in the more complex latent-shape trajectory, with little corresponding benefit for the scalar scale coordinate.
Decoder-consistent endpoint supervision also evaluates the predicted state in
the physical-output subspace retained by the frozen representation without
asking the dynamics model to compensate for a static reconstruction residual.
The matched experiment evaluates the combined multistep training
configuration; it does not attribute the observed gain separately to rollout
exposure, the frozen-decoder target, or the training perturbation.

\hypertarget{efficiency-accuracy-tradeoff}{%
\subsection{Efficiency--accuracy trade-off across aerosol representations}
\label{efficiency-accuracy-tradeoff}}

\begin{figure}[H]
\centering
\includegraphics[width=0.86\textwidth]{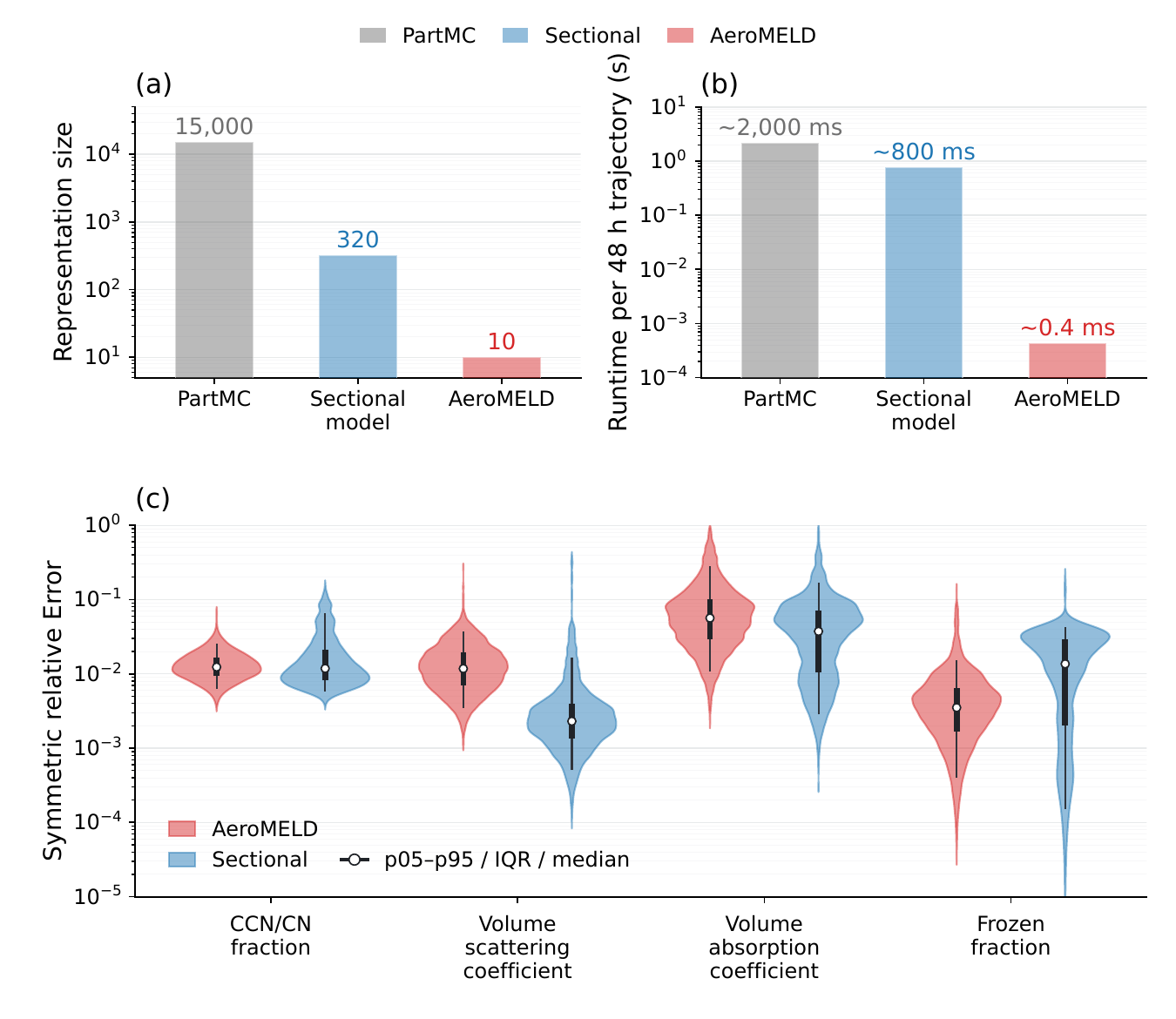}
\caption{Efficiency--accuracy comparison among particle-resolved PartMC, the
evaluated sectional model, and AeroMELD-Coag. (a) Representation size on a
logarithmic scale. PartMC is represented by the nominal
\modelparam{data.nominal_particles}{1{,}000}\(\times\)
\modelparam{data.species_count}{15} particle--species mass entries; the
sectional and AeroMELD values are fixed prognostic coordinates. (b) Mean
integration time per 48-h trajectory with 60-s steps and model output disabled:
single-core CPU time for PyPartMC and the sectional model, and batch-amortized
L40S wall time for AeroMELD. The bars use 2.1866 s, 0.7669 s, and 0.4269 ms,
respectively; annotations show rounded milliseconds. The GPU value is an
amortized throughput measurement, and the different resources preclude an
equal-hardware interpretation. (c) Pooled symmetric relative-error
distributions for CCN-to-CN ratio, volume scattering coefficient, volume
absorption coefficient, and frozen fraction over
\modelparam{split.test_trajectories}{2{,}000} held-out trajectories and all
\modelparam{data.states_per_trajectory}{49} hourly states. Each
trajectory--time pair has equal weight. Red and blue violins denote AeroMELD
and the sectional model, respectively; white circles show medians, thick black
bars show interquartile ranges, and thin black bars span the 5th to 95th
percentiles. The vertical axis is logarithmic. PartMC supplies the
particle-resolved reference and therefore has zero diagnostic error.}
\label{fig:efficiency_accuracy}
\end{figure}

Panel (a) of Figure~\ref{fig:efficiency_accuracy} shows that AeroMELD reduces
the carried state from approximately
\modelparam{baseline.partmc.nominal_particle_state_dim}{15{,}000}
particle--species mass entries in the nominal PartMC population and
\modelparam{baseline.sectional.state_dim}{320} sectional coordinates to
\modelparam{model.total_state_dim}{ten} prognostic coordinates.

Panel (b) places the matched integration times on the same logarithmic scale.
The corresponding per-case throughput is about 5,100 times PyPartMC's and
1,800 times the sectional model's under the reported CPU and GPU
configurations. These ratios do not represent equal-hardware speedups.

Panel (c) shows the diagnostic trade-off. AeroMELD has lower pooled CCN and
frozen-fraction errors, whereas the sectional model has lower scattering and
absorption errors. Across the two reduced models, CCN, scattering, and
frozen-fraction errors cluster around 1\%, and most volume-absorption errors
remain below 10\%. Thus the ten-coordinate state retains useful
mixing-state-sensitive diagnostics while substantially reducing both state
size and integration time in this benchmark.

\hypertarget{discussion-and-conclusion}{%
\section{Discussion and Conclusion}\label{discussion-and-conclusion}}

The held-out evaluation shows that, within the sampled fixed-environment
manifold, the frozen AeroMELD representation serves as both a compact
diagnostic encoding and a useful coordinate for nonlinear coagulation
dynamics. The sectional comparison adds a second result: this learned
coordinate reaches a diagnostic-accuracy regime comparable to the evaluated
conventional reduction with a much smaller prognostic state. The clearest
improvements occur in normalized size-resolved composition and frozen
fraction, and CCN activation has a narrower upper tail. Stable 48-h evolution
of latent shape and total number, coherent decoded behavior, and the three-way
efficiency--accuracy comparison provide empirical support for the present
proof of concept.

The derivation in
Section~\ref{scale-covariant-coagulation-dynamics} separates an exact
structural result from a learned approximation. Binary-coagulation
homogeneity fixes the \(s^2\) and \(s\) factors, whereas closure through
\(\mathbf z\) remains empirical. Encoding the exact factors prevents the
network from spending finite-data capacity on a known concentration law and
extends that law to positive scales outside the sampled training range. Accuracy for unseen latent shapes, environmental conditions, or processes remains to be established. The shared particle map, linear unnormalized population moment, and scale--shape dynamics implement a representation-first scientific-machine-learning strategy. They encode known set structure and operator scaling analytically, leaving only the unresolved closure to be learned.

The linear unnormalized moment also yields a concrete strategy for three-dimensional coupling. In an operator-split implementation, a host model
could advect total number \(N\) and the \modelparam{model.latent_dim}{nine} components of
\(\mathbf m_\phi=N\mathbf z\) as tracer-like prognostic quantities, using the
host model's standard concentration or mixing-ratio convention. Grid-cell
mixing and encoded emission sources would update these additive coordinates
linearly. Each cell would then recover \(\mathbf z=\mathbf m_\phi/N\), apply a
local nonlinear process operator such as the coagulation model studied here,
and convert back to \((N,\mathbf m_\phi)\) before the next transport step.
For the present coagulation operator, the one-step Euler result at 1 h is consistent with the much finer-step solutions, and no stability-driven internal subcycling appears necessary at this coupling interval. The 8 h RK4 comparison also confirms the solver flexibility of the learned continuous-time right-hand side beyond Euler. If combining multiple aerosol processes produces a more nonlinear or stiff latent system, the same state--operator interface could use higher-order integrators such as RK4 or adaptive integrators.
Particle-node graph models retain a more explicit particle description, but
their variable node sets require additional transport and remapping machinery;
the fixed-dimensional moment state offers a different, more direct interface
to Eulerian tracer infrastructure. This coupling pathway follows from the representation structure and awaits demonstration in a coupled model.

Several boundaries remain. Training and evaluation use one sampled scenario
family at fixed temperature and pressure, and distinct normalized
populations can in principle share a latent coordinate while having different
process tendencies. Broader population sampling, environmental conditioning,
and additional process operators are therefore required before coupled
application. The sectional result applies to the evaluated bin
structure, prognostic variables, numerical implementation, and hardware
benchmark. Other sectional schemes may retain different mixing-state detail or
occupy different accuracy--cost regimes. The present comparison establishes a
practical baseline for this experiment. A broader ranking requires evaluations
across additional sectional schemes, environments, and aerosol processes.

The decoded mass-retention diagnostic shows a
\modelparam{results.mass.retention_48_median_pct}{2.2\%} median symmetric drift
error at 48 h and a
\modelparam{results.mass.retention_48_p95_pct}{11.6\%} 95th percentile. The
selected lower- and upper-discrepancy cases in
Appendix~\ref{decoded-total-mass-retention} pair these scalar deviations with
much smaller errors in CCN activation, optical properties, and frozen fraction
for the same rollouts. Budget closure and coupled-model sensitivity remain
unresolved.
Mass evolution in AeroMELD-Coag remains unconstrained by a
coagulation-specific hard projection. Future systems that include emissions,
deposition, condensation, evaporation, and chemistry should apply
process-aware total- or species-mass constraints consistent with the source,
sink, phase-exchange, and conversion budgets of each operator.

For the evaluated setting, AeroMELD-Coag demonstrates that a
\modelparam{model.total_state_dim}{ten}-coordinate state can carry
diagnostically expressive, stable coagulation dynamics with median symmetric
errors of \(\modelparam{results.latent.pooled_median_pct}{3.5\%}\) for latent
shape and \(\modelparam{results.number.pooled_median_pct}{0.52\%}\) for total
number across \modelparam{split.test_trajectories}{2{,}000} held-out
trajectories. Relative to the
\modelparam{baseline.sectional.state_dim}{320}-coordinate sectional baseline,
the AeroMELD state lowers pooled mean composition and frozen-fraction errors and
narrows the CCN upper tail. The sectional model remains more accurate for
binned mass and number and for the optical coefficients. AeroMELD's
batch-amortized L40S integration throughput is about 5,100 times that of
single-core PyPartMC and 1,800 times that of the single-core sectional model
in the matched benchmark. These hardware-specific values do not
represent equal-hardware acceleration. The accuracy and efficiency results show how
particle-resolved simulation can serve as a source for learning compact
prognostic variables that
bridge particle-level diagnostic fidelity and the efficiency required for
future three-dimensional, multi-process coupling.

\section*{Acknowledgements}

This work was supported by the U.S. Department of Energy, Office of Science,
Office of Biological and Environmental Research under Award Number
DE-SC0022130, and the Laboratory Directed Research and Development program at
Sandia National Laboratories. Sandia National Laboratories is a multimission
laboratory managed and operated by National Technology and Engineering
Solutions of Sandia LLC, a wholly owned subsidiary of Honeywell International
Inc. for the U.S. Department of Energy's National Nuclear Security
Administration contract DE-NA0003525. SAND2026-262570.

This paper describes objective technical results and analysis. Any subjective
views or opinions that might be expressed in the paper do not necessarily
represent the views of the U.S. Department of Energy or the United States
Government.

\clearpage
\bibliographystyle{unsrtnat-initials}
\bibliography{references}

@article{poschl2005,
  author  = {P{\"o}schl, Ulrich},
  title   = {Atmospheric aerosols: composition, transformation, climate and health effects},
  journal = {Angewandte Chemie International Edition},
  year    = {2005},
  volume  = {44},
  number  = {46},
  pages   = {7520--7540},
  doi     = {10.1002/anie.200501122}
}

@article{mcfiggans2006,
  author  = {McFiggans, G. and Artaxo, P. and Baltensperger, U. and Coe, H. and Facchini, M. C. and Feingold, G. and Fuzzi, S. and Gysel, M. and Laaksonen, A. and Lohmann, U. and Mentel, T. F. and Murphy, D. M. and O'Dowd, C. D. and Snider, J. R. and Weingartner, E.},
  title   = {The effect of physical and chemical aerosol properties on warm cloud droplet activation},
  journal = {Atmospheric Chemistry and Physics},
  year    = {2006},
  volume  = {6},
  number  = {9},
  pages   = {2593--2649},
  doi     = {10.5194/acp-6-2593-2006}
}

@article{hoose2012,
  author  = {Hoose, C. and M{\"o}hler, O.},
  title   = {Heterogeneous ice nucleation on atmospheric aerosols: A review of results from laboratory experiments},
  journal = {Atmospheric Chemistry and Physics},
  year    = {2012},
  volume  = {12},
  number  = {20},
  pages   = {9817--9854},
  doi     = {10.5194/acp-12-9817-2012}
}

@article{fierce2016,
  author  = {Fierce, Laura and Bond, Tami C. and Bauer, Susanne E. and Mena, Francisco and Riemer, Nicole},
  title   = {Black carbon absorption at the global scale is affected by particle-scale diversity in composition},
  journal = {Nature Communications},
  year    = {2016},
  volume  = {7},
  pages   = {12361},
  doi     = {10.1038/ncomms12361}
}

@article{ching2017,
  author  = {Ching, Joseph and Fast, Jerome and West, Matthew and Riemer, Nicole},
  title   = {Metrics to quantify the importance of mixing state for {CCN} activity},
  journal = {Atmospheric Chemistry and Physics},
  year    = {2017},
  volume  = {17},
  number  = {12},
  pages   = {7445--7458},
  doi     = {10.5194/acp-17-7445-2017}
}

@article{riemer2019mixing,
  author  = {Riemer, N. and Ault, A. P. and West, M. and Craig, R. L. and Curtis, J. H.},
  title   = {Aerosol mixing state: Measurements, modeling, and impacts},
  journal = {Reviews of Geophysics},
  year    = {2019},
  volume  = {57},
  number  = {2},
  pages   = {187--249},
  doi     = {10.1029/2018RG000615}
}

@article{yao2022,
  author  = {Yao, Yu and Curtis, Jeffrey H. and Ching, Joseph and Zheng, Zhonghua and Riemer, Nicole},
  title   = {Quantifying the effects of mixing state on aerosol optical properties},
  journal = {Atmospheric Chemistry and Physics},
  year    = {2022},
  volume  = {22},
  number  = {14},
  pages   = {9265--9282},
  doi     = {10.5194/acp-22-9265-2022}
}

@article{tang2026freezing,
  author  = {Tang, Wenhan and Arabas, Sylwester and Curtis, Jeffrey H. and Knopf, Daniel A. and West, Matthew and Riemer, Nicole},
  title   = {The impact of aerosol mixing state on immersion freezing: insights from classical nucleation theory and particle-resolved simulations},
  journal = {Atmospheric Chemistry and Physics},
  year    = {2026},
  volume  = {26},
  number  = {12},
  pages   = {9221--9255},
  doi     = {10.5194/acp-26-9221-2026}
}

@article{riemer2009partmc,
  author  = {Riemer, Nicole and West, Matthew and Zaveri, Rahul A. and Easter, Richard C.},
  title   = {Simulating the evolution of soot mixing state with a particle-resolved aerosol model},
  journal = {Journal of Geophysical Research: Atmospheres},
  year    = {2009},
  volume  = {114},
  number  = {D9},
  pages   = {D09202},
  doi     = {10.1029/2008JD011073}
}

@article{zaveri2008mosaic,
  author  = {Zaveri, Rahul A. and Easter, Richard C. and Fast, Jerome D. and Peters, Leonard K.},
  title   = {Model for Simulating Aerosol Interactions and Chemistry ({MOSAIC})},
  journal = {Journal of Geophysical Research: Atmospheres},
  year    = {2008},
  volume  = {113},
  number  = {D13},
  pages   = {D13204},
  doi     = {10.1029/2007JD008782}
}

@article{binkowski1995rpm,
  author  = {Binkowski, Francis S. and Shankar, Uma},
  title   = {The regional particulate matter model: 1. Model description and preliminary results},
  journal = {Journal of Geophysical Research: Atmospheres},
  year    = {1995},
  volume  = {100},
  number  = {D12},
  pages   = {26191--26209},
  doi     = {10.1029/95JD02093}
}

@article{whitby1997modal,
  author  = {Whitby, Evan R. and McMurry, Peter H.},
  title   = {Modal aerosol dynamics modeling},
  journal = {Aerosol Science and Technology},
  year    = {1997},
  volume  = {27},
  number  = {6},
  pages   = {673--688},
  doi     = {10.1080/02786829708965504}
}

@article{gelbard1980sectional,
  author  = {Gelbard, Fred and Tambour, Yoram and Seinfeld, John H.},
  title   = {Sectional representations for simulating aerosol dynamics},
  journal = {Journal of Colloid and Interface Science},
  year    = {1980},
  volume  = {76},
  number  = {2},
  pages   = {541--556},
  doi     = {10.1016/0021-9797(80)90394-X}
}

@article{vignati2004m7,
  author  = {Vignati, Elisabetta and Wilson, Julian and Stier, Philip},
  title   = {{M7}: An efficient size-resolved aerosol microphysics module for large-scale aerosol transport models},
  journal = {Journal of Geophysical Research: Atmospheres},
  year    = {2004},
  volume  = {109},
  number  = {D22},
  pages   = {D22202},
  doi     = {10.1029/2003JD004485}
}

@article{bauer2008matrix,
  author  = {Bauer, S. E. and Wright, D. L. and Koch, D. and Lewis, E. R. and McGraw, R. and Chang, L.-S. and Schwartz, S. E. and Ruedy, R.},
  title   = {{MATRIX} ({Multiconfiguration Aerosol TRacker of mIXing state}): an aerosol microphysical module for global atmospheric models},
  journal = {Atmospheric Chemistry and Physics},
  year    = {2008},
  volume  = {8},
  number  = {20},
  pages   = {6003--6035},
  doi     = {10.5194/acp-8-6003-2008}
}

@article{wang2022coagulation,
  author  = {Wang, Justin L. and Curtis, Jeffrey H. and Riemer, Nicole and West, Matthew},
  title   = {Learning coagulation processes with combinatorial neural networks},
  journal = {Journal of Advances in Modeling Earth Systems},
  year    = {2022},
  volume  = {14},
  number  = {12},
  pages   = {e2022MS003252},
  doi     = {10.1029/2022MS003252}
}

@article{ferracina2025gnn,
  author  = {Ferracina, Fabiana and Beeler, Payton and Halappanavar, Mahantesh and Krishnamoorthy, Bala and Minutoli, Marco and Fierce, Laura},
  title   = {Learning to simulate aerosol dynamics with graph neural networks},
  journal = {ACS ES\&T Air},
  year    = {2025},
  volume  = {2},
  number  = {8},
  pages   = {1426--1438},
  doi     = {10.1021/acsestair.4c00261}
}

@article{saleh2025generative,
  author  = {Saleh, Ehsan and Ghaffari, Saba and Curtis, Jeffrey H. and Patel, Lekha and Bosler, Peter A. and Riemer, Nicole and West, Matthew},
  title   = {Compact aerosol size-composition representations preserving {CCN}, optical, and ice nucleation properties},
  journal = {Aerosol Science and Technology},
  year    = {2026},
  volume  = {60},
  number  = {10},
  pages   = {1197--1217},
  doi     = {10.1080/02786826.2026.2695176}
}

@article{saleh2025partial,
  author  = {Saleh, Ehsan and Ghaffari, Saba and Curtis, Jeffrey H. and Patel, Lekha and Bosler, Peter A. and Riemer, Nicole and West, Matthew},
  title   = {Reconstructing the aerosol state from partial observations with generative modeling},
  journal = {Artificial Intelligence for the Earth Systems},
  year    = {2026},
  volume  = {5},
  number  = {4},
  doi     = {10.1175/AIES-D-25-0101.1}
}

@article{saleh2026aeromeld,
  author  = {Saleh, Ehsan and Ghaffari, Saba and Tang, Wenhan and Curtis, Jeffrey H. and Patel, Lekha and Bosler, Peter A. and Riemer, Nicole and West, Matthew},
  title   = {{AeroMELD}: A linear embedding of aerosol populations for diagnostics and latent dynamics},
  journal = {arXiv preprint arXiv:2607.11073},
  year    = {2026},
  doi     = {10.48550/arXiv.2607.11073}
}

@inproceedings{zaheer2017deepsets,
  author    = {Zaheer, Manzil and Kottur, Satwik and Ravanbakhsh, Siamak and P{\'o}czos, Barnab{\'a}s and Salakhutdinov, Ruslan and Smola, Alexander J.},
  title     = {Deep Sets},
  booktitle = {Advances in Neural Information Processing Systems},
  volume    = {30},
  year      = {2017}
}

@inproceedings{chen2018neuralode,
  author    = {Chen, Ricky T. Q. and Rubanova, Yulia and Bettencourt, Jesse and Duvenaud, David K.},
  title     = {Neural ordinary differential equations},
  booktitle = {Advances in Neural Information Processing Systems},
  volume    = {31},
  pages     = {6571--6583},
  year      = {2018}
}

@book{villani2009optimal,
  author    = {Villani, C{\'e}dric},
  title     = {Optimal Transport: Old and New},
  series    = {Grundlehren der mathematischen Wissenschaften},
  volume    = {338},
  publisher = {Springer},
  address   = {Berlin, Heidelberg},
  year      = {2009},
  doi       = {10.1007/978-3-540-71050-9}
}

@article{gasparik2020scenario,
  author  = {Gasparik, J. T. and Ye, Q. and Curtis, J. H. and Presto, A. A. and Donahue, N. M. and Sullivan, R. C. and West, M. and Riemer, N.},
  title   = {Quantifying errors in the aerosol mixing-state index based on limited particle sample size},
  journal = {Aerosol Science and Technology},
  year    = {2020},
  volume  = {54},
  number  = {12},
  pages   = {1527--1541},
  doi     = {10.1080/02786826.2020.1804523}
}

@article{riemer2010aging,
  author  = {Riemer, Nicole and West, Matthew and Zaveri, Rahul and Easter, Richard},
  title   = {Estimating black carbon aging time-scales with a particle-resolved aerosol model},
  journal = {Journal of Aerosol Science},
  year    = {2010},
  volume  = {41},
  number  = {1},
  pages   = {143--158},
  doi     = {10.1016/j.jaerosci.2009.08.009}
}

@article{niemand2012inas,
  author  = {Niemand, Monika and M{\"o}hler, Ottmar and Vogel, Bernhard and Vogel, Heike and Hoose, Corinna and Connolly, Paul and Klein, Holger and Bingemer, Heinz and DeMott, Paul and Skrotzki, Julian and Leisner, Thomas},
  title   = {A particle-surface-area-based parameterization of immersion freezing on desert dust particles},
  journal = {Journal of the Atmospheric Sciences},
  year    = {2012},
  volume  = {69},
  number  = {10},
  pages   = {3077--3092},
  doi     = {10.1175/JAS-D-11-0249.1}
}

@article{schill2020blackcarbon,
  author  = {Schill, Gregory P. and DeMott, Paul J. and Emerson, Ethan W. and Rauker, Anne Marie C. and Kodros, John K. and Suski, Kaitlyn J. and Hill, Thomas C. J. and Levin, Ezra J. T. and Pierce, Jeffrey R. and Farmer, Delphine K. and Kreidenweis, Sonia M.},
  title   = {The contribution of black carbon to global ice nucleating particle concentrations relevant to mixed-phase clouds},
  journal = {Proceedings of the National Academy of Sciences},
  year    = {2020},
  volume  = {117},
  number  = {37},
  pages   = {22705--22711},
  doi     = {10.1073/pnas.2001674117}
}

@article{brenowitz2018prognostic,
  author  = {Brenowitz, N. D. and Bretherton, C. S.},
  title   = {Prognostic Validation of a Neural Network Unified Physics Parameterization},
  journal = {Geophysical Research Letters},
  volume  = {45},
  number  = {12},
  pages   = {6289--6298},
  year    = {2018},
  doi     = {10.1029/2018GL078510}
}

@inproceedings{loshchilov2019decoupled,
  author    = {Loshchilov, Ilya and Hutter, Frank},
  title     = {Decoupled Weight Decay Regularization},
  booktitle = {International Conference on Learning Representations},
  year      = {2019}
}

\clearpage
\appendix

\hypertarget{supporting-scale-shape-results}{%
\section{Supporting scale--shape results}\label{supporting-scale-shape-results}}

\begin{proposition}[No recurrence of the complete coagulation state]
If \(N>0\), \(K\ge0\), and \(\iint K\,d\widehat\mu d\widehat\mu>0\) along an interval, then \(N\) decreases strictly and the complete state \((N,\mathbf z)\) cannot recur on that interval.
\end{proposition}

\begin{proof}
Combining the first identity in Eq.~\eqref{eq:scale-shape-functionals} with the definition of \(a(\widehat\mu)\) in Eq.~\eqref{eq:normalized-tendency-functionals} gives \(dN/dt<0\) under the stated conditions. The complete state cannot revisit an earlier value because
that would require the same value of \(N\) at two distinct times.
\end{proof}

The projected curve \(\mathbf z(t)\) may still bend, cross itself, or return to a previous projected value at a different \(N\).

\begin{proposition}[Regularity inherited by the latent path]
If each coordinate \(\phi_j\) is Lipschitz with constant \(L_j\), then
\begin{equation}
\|\mathbf z(t_2)-\mathbf z(t_1)\|_2
\le\left(\sum_jL_j^2\right)^{1/2}
W_1(\widehat\mu_{t_2},\widehat\mu_{t_1}).
\label{eq:appendix-wasserstein-vector}
\end{equation}
\end{proposition}

\begin{proof}
Kantorovich--Rubinstein duality gives \(|\int\phi_jd\widehat\mu_{t_2}-\int\phi_jd\widehat\mu_{t_1}|\le L_jW_1(\widehat\mu_{t_2},\widehat\mu_{t_1})\) for every coordinate \cite{villani2009optimal}. Squaring, summing over coordinates, and taking the square root yields Eq.~\eqref{eq:appendix-wasserstein-vector}.
\end{proof}

\hypertarget{frozen-decoder-target-properties}{%
\section{Frozen-decoder target properties}\label{frozen-decoder-target-properties}}

\begin{lemma}[Residual neutrality of the frozen-decoder reference target]
\label{lem:residual-neutrality}
If \(\widehat{\mathbf z}_{i,k+h}=\mathbf z^\star_{i,k+h}\), then \(\ell_{\mathrm{FDT},i,k}^{(h)}=0\), independently of the frozen Encoder--Decoder residual relative to the raw physical state.
\end{lemma}

\begin{proof}
Equality of the two latent endpoints gives identical inputs to every block of the same deterministic frozen Decoder. Their forward values are therefore equal before the stop-gradient operation, which changes the backward graph but not the reference value. Every term in Eq.~\eqref{eq:frozen-decoder-target} is zero. With a direct raw-physical target, the frozen representation residual remains, and the loss may be nonzero at the reference latent state.
\end{proof}

For the local interpretation, write \(\widehat{\mathbf z}=\mathbf z^\star+\boldsymbol\delta\) and let \(J_q(\mathbf z^\star)\) be the Jacobian of Decoder block \(D_{u,q}\). If the Decoder is twice continuously differentiable in a neighborhood of \(\mathbf z^\star\), then
\begin{equation}
\ell_{\mathrm{FDT},i,k}^{(h)}
=\boldsymbol\delta^\top
\left[
\frac1{|\mathcal Q|}\sum_{q\in\mathcal Q}\frac1{d_q}
J_q(\mathbf z^\star)^\top J_q(\mathbf z^\star)
\right]
\boldsymbol\delta
+\mathcal O(\|\boldsymbol\delta\|_2^3).
\label{eq:fdt-local-metric}
\end{equation}
The bracketed matrix is positive semidefinite and weights latent errors by local Decoder sensitivity. It need not be positive definite because the Decoder is not assumed to be locally injective.

\hypertarget{architecture-transformations-and-training-objective}{%
\section{Architecture, transformations, and training objective}\label{architecture-transformations-and-training-objective}}

\setcounter{table}{0}
\renewcommand{\thetable}{C\arabic{table}}
\renewcommand{\theHtable}{C.\arabic{table}}

\begin{table}[H]
\centering
\caption{Model architecture used in the selected experiment.}
\label{tab:model-architecture}
\small
\begin{tabularx}{\textwidth}{@{}>{\raggedright\arraybackslash}X>{\raggedright\arraybackslash}X@{}}
\toprule
Component & Configuration \\
\midrule
Particle state & \modelparam{data.species_count}{15} species masses plus represented-number weights \\
Particle network & \modelparam{encoder.hidden_layers}{Two} \modelparam{encoder.hidden_width}{128}-unit ReLU hidden layers \\
Set aggregation & Number-weighted mean \\
Latent shape & Deterministic mean, \(d_z=\modelparam{model.latent_dim}{9}\) \\
Decoder & \modelparam{decoder.hidden_layers}{Two} \modelparam{decoder.hidden_width}{128}-unit ReLU hidden layers; \modelparam{decoder.output_dim}{550} outputs \\
Latent-dynamics input & \(\mathbf z\in\mathbb R^{\modelparam{model.latent_dim}{9}}\) \\
Latent-dynamics network & \modelparam{dynamics.hidden_layers}{Two} \modelparam{dynamics.hidden_width}{64}-unit ReLU layers; dropout \modelparam{dynamics.dropout}{0.1} \\
Joint output & One affine \modelparam{model.total_state_dim}{10}-output layer split into \modelparam{model.latent_dim}{9} shape-rate and 1 scale-rate components \\
Latent-dynamics parameters & \modelparam{dynamics.parameter_count}{5{,}450} \\
Integrator & Forward Euler, \(\Delta t=\modelparam{training.integration_step_h}{1}\) h \\
\bottomrule
\end{tabularx}
\end{table}

\begin{table}[H]
\centering
\caption{Diagnostic-output and rate transformations.}
\label{tab:diagnostic-transformations}
\small
\begin{tabularx}{\textwidth}{@{}>{\raggedright\arraybackslash}Xccc@{}}
\toprule
Block & Dimension & Divide by \(N\) & Signed-power exponent \\
\midrule
Binned mass \(M_b\) & \modelparam{decoder.mass_bins}{20} & Yes & \modelparam{transform.mass_power}{0.3} \\
Normalized composition \(P_{ab}\) & \(\modelparam{decoder.mass_bins}{20}\times\modelparam{decoder.composition_species}{15}\) & No & \modelparam{transform.composition_power}{0.5} \\
Binned number \(N_b\) & \modelparam{decoder.mass_bins}{20} & Yes & \modelparam{transform.number_power}{0.5} \\
CCN-to-CN ratio & \modelparam{decoder.ccn_points}{100} & No & \modelparam{transform.ccn_power}{1.0} \\
Scattering & \modelparam{decoder.optical_points}{5} & Yes & \modelparam{transform.scattering_power}{0.1} \\
Absorption & \modelparam{decoder.optical_points}{5} & Yes & \modelparam{transform.absorption_power}{0.1} \\
Frozen fraction & \modelparam{decoder.frozen_points}{100} & No & \modelparam{transform.frozen_power}{0.1} \\
Latent rate & \modelparam{model.latent_dim}{9} & Not applicable & \modelparam{transform.latent_rate_power}{1.0} \\
Scale rate & 1 & Not applicable & \modelparam{transform.scale_rate_power}{1.0} \\
\bottomrule
\end{tabularx}
\end{table}

Each transformed block is standardized using statistics fitted on the training partition. The fitted parameters remain fixed during training and evaluation and are also used in the corresponding inverse transformations.

\begin{table}[H]
\centering
\caption{Latent-dynamics training objective and optimization.}
\label{tab:training-objective}
\small
\begin{tabularx}{\textwidth}{@{}>{\raggedright\arraybackslash}X>{\raggedright\arraybackslash}X@{}}
\toprule
Setting & Value \\
\midrule
Horizons and windows & \modelparam{training.horizon_min_h}{1}--\modelparam{training.horizon_max_h}{12} h; all eligible stride-one windows \\
Integration step & \modelparam{training.integration_step_h}{1} h \\
Rate objective & MSE of predicted and corrected-reference window-mean transformed rates \\
Decoded objective & Equal-block endpoint FDT \\
Horizon weighting & Equal total training weight across horizons; final partial batches retained and normalized by their actual window count \\
Family weighting & Latent-rate, scale-rate, and complete endpoint FDT families each \(\modelparam{training.family_weight}{1/3}\); direct state-rollout weights are \modelparam{training.direct_rollout_weight}{zero} \\
FDT blocks & \modelparam{decoder.block_count}{Seven} blocks, equally weighted \\
Training perturbation & One initial Gaussian draw, scale \modelparam{training.perturbation_scale}{0.1} in zero-centered transformed latent-rate coordinates \\
Target correction & Exact subtraction from the first latent-rate target \\
Perturbation lifecycle & Disabled for validation, test, and simulation \\
Optimizer & AdamW; learning rate \(\modelparam{training.learning_rate}{10^{-4}}\); weight decay \(\modelparam{training.weight_decay}{10^{-4}}\) \\
Batch and clipping & \modelparam{training.batch_size}{64}; gradient-norm clip \modelparam{training.gradient_clip}{1.0} \\
\bottomrule
\end{tabularx}
\end{table}

\hypertarget{evaluation-and-supplementary-results}{%
\section{Evaluation and supplementary results}\label{evaluation-and-supplementary-results}}

\setcounter{table}{0}
\renewcommand{\thetable}{D\arabic{table}}
\renewcommand{\theHtable}{D.\arabic{table}}
\setcounter{figure}{0}
\renewcommand{\thefigure}{D\arabic{figure}}
\renewcommand{\theHfigure}{D.\arabic{figure}}

\hypertarget{diagnostic-evaluation-domains}{%
\subsection{Diagnostic evaluation domains}\label{diagnostic-evaluation-domains}}

\begin{table}[H]
\centering
\caption{Diagnostic domains used in the symmetric-error calculation.}
\label{tab:evaluation-domains}
\small
\begin{tabularx}{\textwidth}{@{}>{\raggedright\arraybackslash}X>{\raggedright\arraybackslash}X>{\raggedright\arraybackslash}X@{}}
\toprule
Quantity & Evaluation domain & Pre-metric transformation \\
\midrule
\(\mathbf z\) & All 9 coordinates & None \\
\(N\) & Physical total number & None \\
\(M_b\) & All 20 size bins & None \\
\(P_{ab}\) & All \(20\times15\) entries & None; Frobenius norm \\
\(N_b\) & All 20 size bins & None \\
CCN-to-CN & Supersaturation \(0.1\)--\(0.6\%\) & None \\
Scattering & 300--1,000 nm & \(\mathrm m^{-1}\) to \(\mathrm{Mm}^{-1}\); \(\log_{10}[\max(q,0.1)]\) \\
Absorption & 300--1,000 nm & \(\mathrm m^{-1}\) to \(\mathrm{Mm}^{-1}\); \(\log_{10}[\max(q,0.1)]\) \\
Frozen fraction & \(-25\) to \(-10\ ^\circ\mathrm C\) & \(\log_{10}(q+10^{-30})\) \\
\bottomrule
\end{tabularx}
\end{table}

\hypertarget{illustrative-rollouts-across-error-distribution}{%
\subsection{Illustrative rollouts across the held-out error distribution}
\label{illustrative-rollouts-across-error-distribution}}

To complement the 50th-percentile trajectory in
Figure~\ref{fig:coag_rollout_example},
Figures~\ref{fig:coag_rollout_lower_error}
and~\ref{fig:coag_rollout_upper_tail} show held-out trajectories at the 10th-
and 90th-percentile positions in the same evaluation-only composite MSE ranking.
The case-selection score gives equal weight to latent shape, total number, and
the seven decoded quantity groups and averages across the evaluated times.
These percentile labels describe relative positions in the empirical ranking: the 10th-percentile case is a lower-error example, and the 90th-percentile case characterizes upper-tail behavior below the maximum held-out error.

\begin{figure}[H]
\centering
\includegraphics[width=\textwidth]{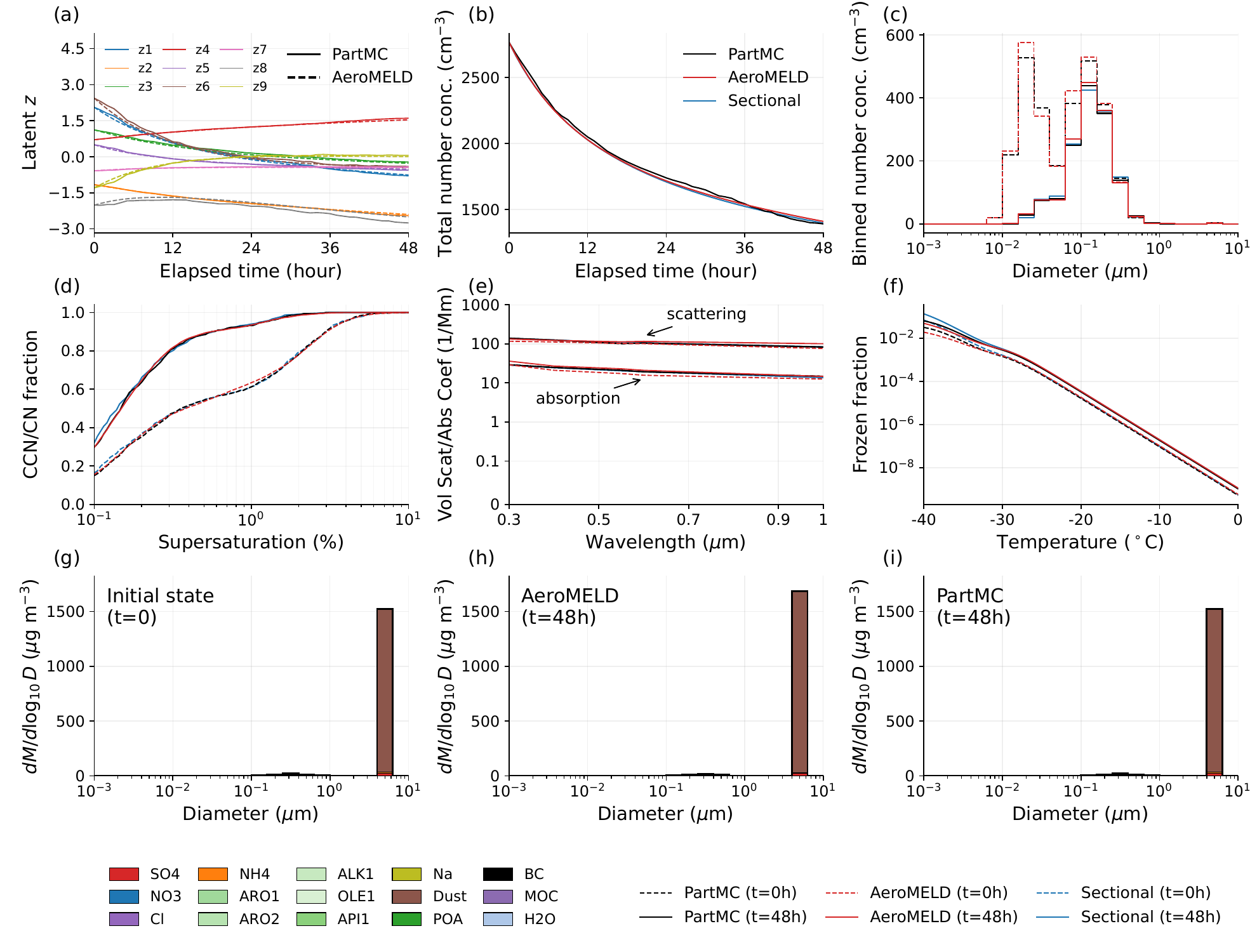}
\caption{Illustrative 48-h coagulation rollout for the held-out case at the
10th-percentile position in the same evaluation-only composite MSE ranking
used to select the 50th-percentile case in
Figure~\ref{fig:coag_rollout_example}; lower percentiles indicate smaller
composite errors. (a) Nine latent-shape coordinates; solid curves denote
encoded PartMC reference states and dashed curves denote the latent-dynamics
rollout. (b) Physical total number concentration, with the reference in black,
the prediction in red, and the sectional-model result in blue. (c--f) Binned
number concentration, CCN-to-CN ratio, aerosol optical properties, and frozen
fraction; dashed and solid curves denote the initial and 48-h states,
respectively, while black, red, and blue denote reference, prediction, and
sectional-model results. Arrows distinguish scattering and absorption
in the optical panel. (g--i) Size-resolved composition at the initial reference
state, predicted 48-h state, and reference 48-h state. Black outlines show
total binned mass \(M_b\), and colors show normalized species fractions
\(P_{ab}\).}
\label{fig:coag_rollout_lower_error}
\end{figure}

\begin{figure}[H]
\centering
\includegraphics[width=\textwidth]{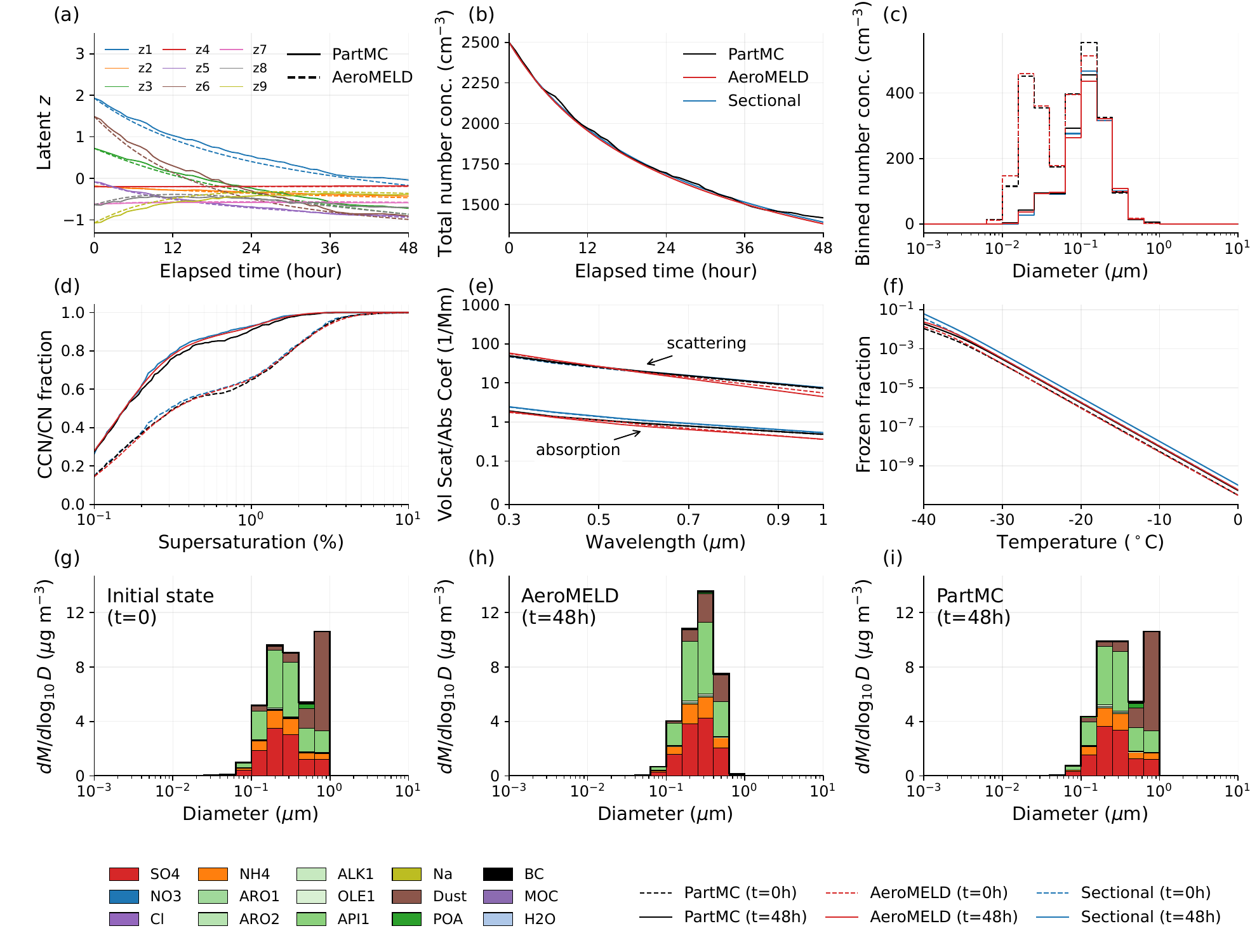}
\caption{Illustrative 48-h coagulation rollout for the held-out case at the
90th-percentile position in the same evaluation-only composite MSE ranking
used to select the 50th-percentile case in
Figure~\ref{fig:coag_rollout_example}. This position characterizes a relatively
high-error rollout below the maximum held-out error. (a) Nine
latent-shape coordinates; solid curves denote encoded PartMC reference states
and dashed curves denote the latent-dynamics rollout. (b) Physical total
number concentration, with the reference in black, the prediction in red, and
the sectional-model result in blue. (c--f) Binned number concentration,
CCN-to-CN ratio, aerosol optical properties, and frozen fraction; dashed and
solid curves denote the initial and 48-h states, respectively, while black,
red, and blue denote reference, prediction, and sectional-model results. Arrows
distinguish scattering and absorption in the optical
panel. (g--i) Size-resolved composition at the initial reference state,
predicted 48-h state, and reference 48-h state. Black outlines show total
binned mass \(M_b\), and colors show normalized species fractions
\(P_{ab}\).}
\label{fig:coag_rollout_upper_tail}
\end{figure}

Across the three percentile-ranked examples, total number decreases and the
size distribution shifts toward larger particles. Visible discrepancies
increase toward the upper-tail case, especially in the total-number path,
binned number distribution, and optical and CCN responses. These examples connect empirical ranking position to qualitative rollout behavior; the population-level error distributions in Figures~\ref{fig:error_distributions} and~\ref{fig:error_vs_time} remain the basis for aggregate evaluation.

\hypertarget{ed-reconstruction-and-ld-discrepancy}{%
\subsection{Separating Encoder--Decoder reconstruction error from
latent-dynamics rollout discrepancy}\label{ed-reconstruction-and-ld-discrepancy}}

The model-space comparison separates three transformed decoded states at each
held-out trajectory and time: raw transformed truth \(\mathbf u^\star\), the
frozen Encoder--Decoder reconstruction \(D_u(\mathbf z^\star)\), and the
decoded latent-dynamics rollout \(D_u(\widehat{\mathbf z})\). We refer to the
difference between the first two as the Encoder--Decoder (ED) reconstruction
residual and to the additional difference produced by the rollout as the
latent-dynamics (LD) discrepancy:
\begin{equation}
\mathbf r_{\mathrm{ED}}=D_u(\mathbf z^\star)-\mathbf u^\star,
\qquad
\mathbf r_{\mathrm{LD}}=D_u(\widehat{\mathbf z})-D_u(\mathbf z^\star),
\qquad
\mathbf r_{\mathrm{total}}=\mathbf r_{\mathrm{ED}}+\mathbf r_{\mathrm{LD}}.
\label{eq:ed-ld-residuals}
\end{equation}
Because the decoded blocks are standardized on the training partition, Figure~\ref{fig:ed_ld_decomposition} reports root-mean-square residuals after equal averaging within each block and across the \modelparam{decoder.block_count}{seven} blocks. This metric is aligned with the transformed Decoder output and is distinct from the raw-physical symmetric error used in Figures~\ref{fig:error_distributions}--\ref{fig:error_vs_time}.

In the pooled comparison, reconstruction residuals exceed rollout-induced discrepancies by a factor of \modelparam{results.ed_ld.pooled_ratio}{4.9} and are larger in all \modelparam{decoder.block_count}{seven} decoded blocks. The reconstruction-to-rollout ratio is also at least three at 12, 24, and 48 h. Thus the frozen representation sets the principal transformed-output floor in the pooled diagnostic and at these preregistered lead times, with a smaller additional discrepancy from the learned dynamics.

Squared total error also contains the cross term \(2\langle\mathbf r_{\mathrm{ED}},\mathbf r_{\mathrm{LD}}\rangle\). We compute and report that term in the accompanying data artifact. The residual-magnitude ratios compare magnitudes; additive percentages of total error also require the cross term.

\begin{figure}[H]
\centering
\includegraphics[width=\textwidth]{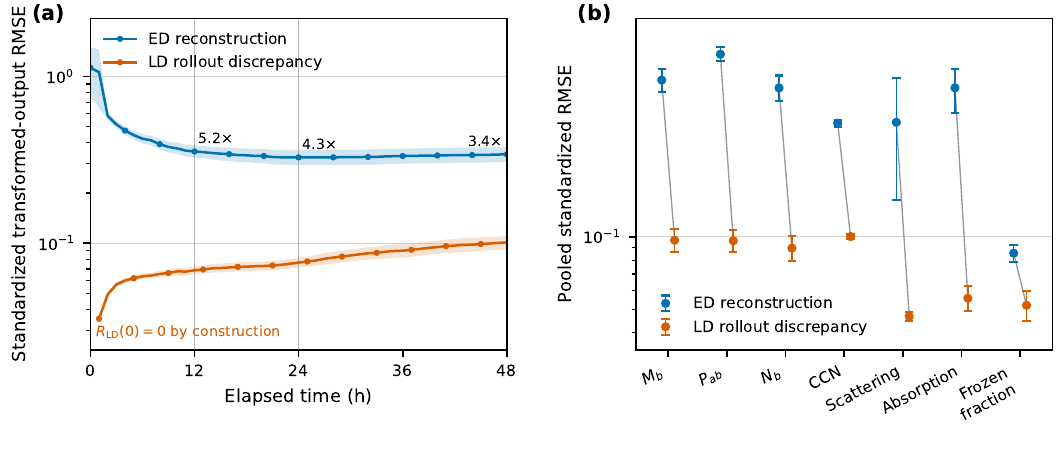}
\caption{Frozen-representation reconstruction error and additional latent-rollout discrepancy for the final model on all \modelparam{split.test_trajectories}{2{,}000} held-out trajectories. (a) Aggregate root-mean-square residual in standardized transformed-output space at each hourly state; curves average equally within each decoded block and then across the \modelparam{decoder.block_count}{seven} blocks, and shaded bands are 90\% trajectory-bootstrap intervals. The rollout discrepancy is zero at \(t=0\) by construction. (b) Pooled root-mean-square residual for every decoded block across trajectories and post-initial times; markers show point estimates and whiskers show 90\% trajectory-bootstrap intervals. Both panels use the Encoder--Decoder (ED) and latent-dynamics (LD) residual definitions in Eq.~\eqref{eq:ed-ld-residuals}; the ratios are magnitude comparisons, and additive error fractions would also require the cross term.}
\label{fig:ed_ld_decomposition}
\end{figure}

\hypertarget{decoded-total-mass-retention}{%
\subsection{Decoded total-mass retention}\label{decoded-total-mass-retention}}

Binary coagulation conserves total aerosol mass in the PartMC reference.
Figure~\ref{fig:total_mass_retention} therefore evaluates whether the decoded
AeroMELD rollout preserves its initial total mass without a hard projection or
a dedicated mass-conservation penalty. Panel (a) is normalized separately by each decoded trajectory's value at \(t=0\), so it diagnoses temporal drift independently of the frozen Encoder--Decoder's absolute initial-state mass bias. Panel (b) compares decoded total mass directly with the raw PartMC reference and therefore retains both the static representation offset and subsequent rollout drift.
The population summary is followed by two deliberately selected cases from the
lower and upper tails of the extrema of that self-normalized decoded-mass ratio
to characterize the accompanying errors in the other decoded diagnostics.

\begin{figure}[H]
\centering
\includegraphics[width=\textwidth]{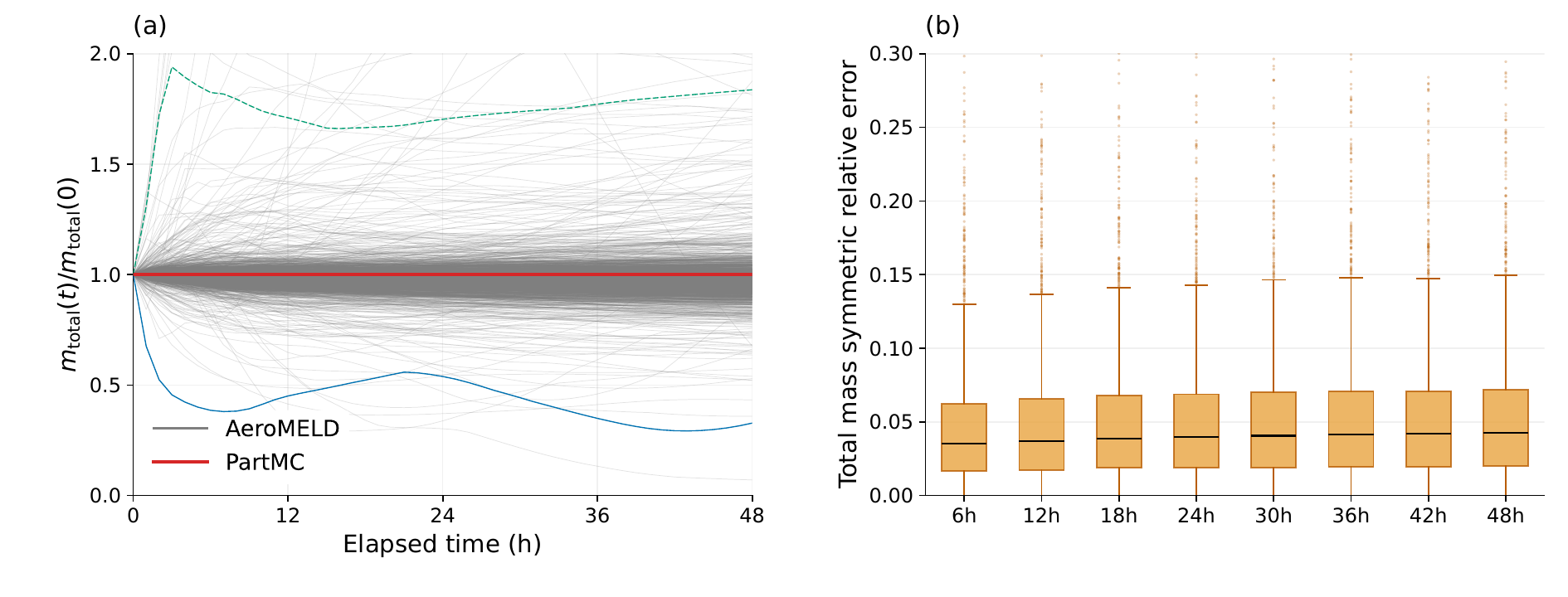}
\caption{Decoded total-mass behavior over
\modelparam{split.test_trajectories}{2{,}000} held-out 48-h coagulation
rollouts. (a) Decoded AeroMELD total aerosol mass normalized by the same
trajectory's value at \(t=0\). Thin gray curves show individual rollouts, and
the red line at unity marks the ideal mass-conserving PartMC baseline. The
solid blue and dashed green curves identify the lower- and upper-extreme cases
examined in Figures~\ref{fig:coag_rollout_lower_mass_discrepancy}
and~\ref{fig:coag_rollout_upper_mass_discrepancy}. (b) Distributions of the
symmetric relative error between decoded total mass and the raw PartMC
total-mass reference at 6-h intervals; these errors include both the frozen-
representation offset and temporal rollout drift. Boxes show medians and
interquartile ranges; whiskers extend to \(1.5\) times the interquartile range,
and points denote outliers. The lower case is the third-smallest trajectory
minimum among complete curves within the displayed 0--2 range; the upper case
has the largest trajectory maximum, in both cases over saved times from 0 to
48 h.}
\label{fig:total_mass_retention}
\end{figure}

\begin{figure}[H]
\centering
\includegraphics[width=\textwidth]{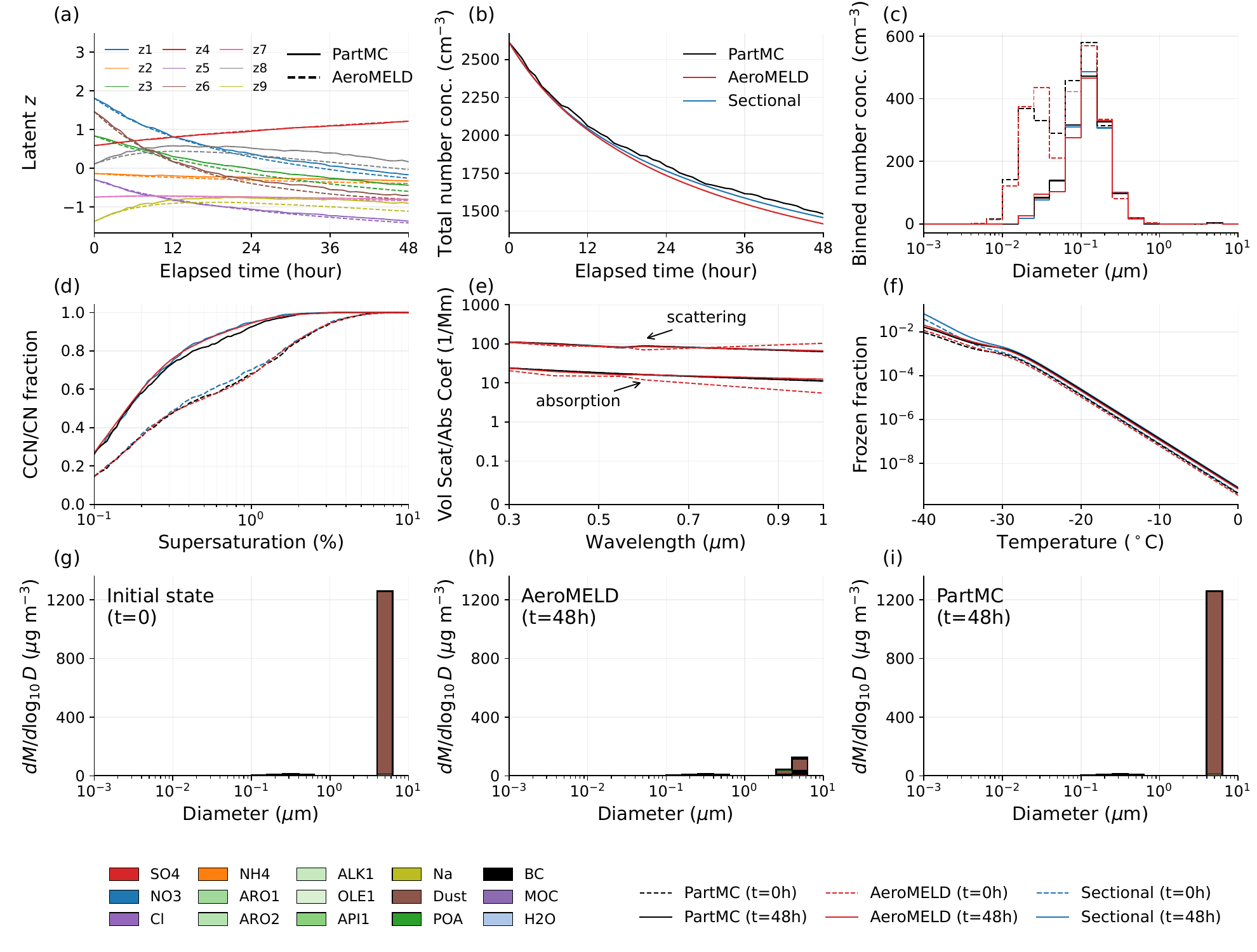}
\caption{Illustrative 48-h rollout selected from the lower tail of the minimum
self-normalized decoded-total-mass ratio, corresponding to the solid blue curve
in Figure~\ref{fig:total_mass_retention}(a). The selected case reaches
\modelparam{results.mass.lower_case.extreme_self_normalized_ratio}{0.292} of
its decoded \(t=0\) mass at
\modelparam{results.mass.lower_case.extreme_time_h}{43} h. (a) Nine
latent-shape coordinates \(z_1\)--\(z_9\); solid curves denote encoded PartMC
reference states and dashed curves denote the AeroMELD rollout. (b) Physical
total number concentration. (c--f) Binned number concentration, CCN-to-CN
ratio, aerosol optical properties, and frozen fraction; black, red, and blue
denote PartMC, AeroMELD, and the sectional result, respectively. (g--i)
Size-resolved composition at the initial and 48-h states; these panels show
PartMC and AeroMELD only. Case selection uses the mass-ratio extreme
independently of the overall composite diagnostic-error ranking.}
\label{fig:coag_rollout_lower_mass_discrepancy}
\end{figure}

\begin{figure}[H]
\centering
\includegraphics[width=\textwidth]{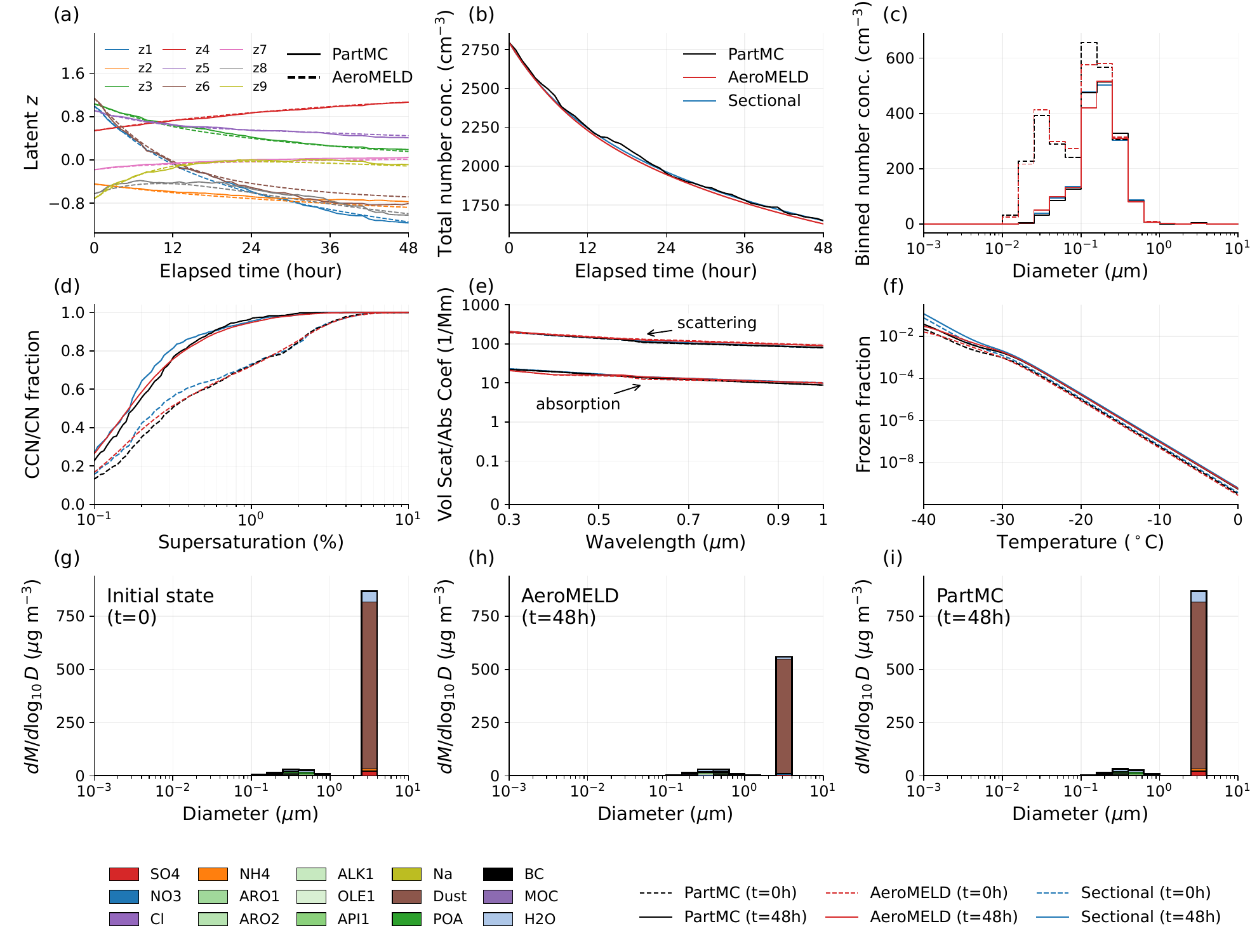}
\caption{Illustrative 48-h rollout selected from the upper tail of the maximum
self-normalized decoded-total-mass ratio, corresponding to the dashed green
curve in Figure~\ref{fig:total_mass_retention}(a). The selected case reaches
\modelparam{results.mass.upper_case.extreme_self_normalized_ratio}{1.94} of
its decoded \(t=0\) mass at
\modelparam{results.mass.upper_case.extreme_time_h}{3} h. Panels and line
conventions follow
Figure~\ref{fig:coag_rollout_lower_mass_discrepancy}, including the blue
sectional results in panels (b--f) and their absence from the composition
panels. Case selection uses the mass-ratio extreme independently of the
overall composite diagnostic-error ranking.}
\label{fig:coag_rollout_upper_mass_discrepancy}
\end{figure}

These cases separate the scalar decoded-mass budget from the multivariate
outputs used to diagnose aerosol--cloud, optical, and freezing behavior. At
48 h, the symmetric decoded-total-mass errors relative to PartMC are
\modelparam{results.mass.lower_case.partmc_error_48_pct}{72.3\%} for the lower
case and \modelparam{results.mass.upper_case.partmc_error_48_pct}{19.1\%} for
the upper case. Across the same two trajectories, the 48-h errors are
at most \modelparam{results.mass.selected_cases.ccn_48_max_pct}{2.23\%} for
CCN, \modelparam{results.mass.selected_cases.scattering_48_max_pct}{0.785\%}
for scattering,
\modelparam{results.mass.selected_cases.absorption_48_max_pct}{1.95\%} for
absorption, and
\modelparam{results.mass.selected_cases.frozen_48_max_pct}{0.537\%} for frozen
fraction. In these two selected rollouts, a large scalar decoded-mass
discrepancy coexists with much smaller errors in the retained diagnostics. The
size--composition structure and diagnostic mapping remain useful in these
individual rollouts despite incomplete scalar mass closure. The relationship
between these diagnostics and aerosol mass beyond the selected cases remains
unresolved.

Mass conservation remains essential for process coupling, source--sink
budgets, and subsequent microphysics. The fixed-environment coagulation
rollouts evaluate diagnostic behavior for the selected cases; cloud formation,
precipitation, radiation, and ice production in a coupled model lie outside
this evaluation. Enforcing or budgeting mass consistently across multiple
process operators remains an important extension.

\hypertarget{multistep-single-step-tables}{%
\subsection{Matched multistep--single-step comparison}\label{multistep-single-step-tables}}

Figure~\ref{fig:single_multi_step} compares per-lead-time symmetric relative
errors through 48 h. The supplementary tables below report a complementary
24-h summary for the same matched models, frozen representation, test
trajectories, and one-hour Euler integration. For each trajectory, overall
latent RMSE is computed over all nine latent coordinates and lead times from
1 to 24 h. We define \(\mathrm{SS/MS}\) as single-step error divided by
multistep error, so values above one favor multistep training.

\begin{table}[H]
\centering
\caption{Distribution of per-trajectory latent RMSE over \modelparam{split.test_trajectories}{2{,}000} held-out 24-h rollouts. SS/MS is single-step error divided by multistep error.}
\label{tab:multistep-overall}
\begin{tabular}{@{}lrrr@{}}
\toprule
Statistic & Multistep & Single-step & SS/MS \\
\midrule
Mean & \modelparam{results.matched.multistep.trajectory_24_mean_rmse}{0.0527} & \modelparam{results.matched.single_step.trajectory_24_mean_rmse}{0.0586} & \modelparam{results.matched.ss_ms.trajectory_24_mean_ratio}{1.11} \\
Median & \modelparam{results.matched.multistep.trajectory_24_median_rmse}{0.0477} & \modelparam{results.matched.single_step.trajectory_24_median_rmse}{0.0523} & \modelparam{results.matched.ss_ms.trajectory_24_median_ratio}{1.10} \\
P90 & \modelparam{results.matched.multistep.trajectory_24_p90_rmse}{0.0807} & \modelparam{results.matched.single_step.trajectory_24_p90_rmse}{0.0903} & \modelparam{results.matched.ss_ms.trajectory_24_p90_ratio}{1.12} \\
P95 & \modelparam{results.matched.multistep.trajectory_24_p95_rmse}{0.0941} & \modelparam{results.matched.single_step.trajectory_24_p95_rmse}{0.1072} & \modelparam{results.matched.ss_ms.trajectory_24_p95_ratio}{1.14} \\
P99 & \modelparam{results.matched.multistep.trajectory_24_p99_rmse}{0.1343} & \modelparam{results.matched.single_step.trajectory_24_p99_rmse}{0.1532} & \modelparam{results.matched.ss_ms.trajectory_24_p99_ratio}{1.14} \\
Maximum & \modelparam{results.matched.multistep.trajectory_24_maximum_rmse}{0.4284} & \modelparam{results.matched.single_step.trajectory_24_maximum_rmse}{0.7445} & \modelparam{results.matched.ss_ms.trajectory_24_maximum_ratio}{1.74} \\
\bottomrule
\end{tabular}
\end{table}

\begin{table}[H]
\centering
\caption{Lead-time dependence of mean latent RMSE across \modelparam{split.test_trajectories}{2{,}000} held-out trajectories. SS/MS is single-step error divided by multistep error.}
\label{tab:multistep-lead-time}
\begin{tabular}{@{}lrrr@{}}
\toprule
Lead time & Multistep & Single-step & SS/MS \\
\midrule
1 h & \modelparam{results.matched.multistep.lead_1_mean_trajectory_rmse}{0.0188} & \modelparam{results.matched.single_step.lead_1_mean_trajectory_rmse}{0.0188} & \modelparam{results.matched.ss_ms.lead_1_mean_trajectory_ratio}{1.004} \\
6 h & \modelparam{results.matched.multistep.lead_6_mean_trajectory_rmse}{0.0418} & \modelparam{results.matched.single_step.lead_6_mean_trajectory_rmse}{0.0439} & \modelparam{results.matched.ss_ms.lead_6_mean_trajectory_ratio}{1.049} \\
12 h & \modelparam{results.matched.multistep.lead_12_mean_trajectory_rmse}{0.0518} & \modelparam{results.matched.single_step.lead_12_mean_trajectory_rmse}{0.0571} & \modelparam{results.matched.ss_ms.lead_12_mean_trajectory_ratio}{1.104} \\
24 h & \modelparam{results.matched.multistep.lead_24_mean_trajectory_rmse}{0.0625} & \modelparam{results.matched.single_step.lead_24_mean_trajectory_rmse}{0.0734} & \modelparam{results.matched.ss_ms.lead_24_mean_trajectory_ratio}{1.175} \\
\bottomrule
\end{tabular}
\end{table}

Multistep training lowers per-trajectory error in \modelparam{results.matched.multistep_improved_fraction_pct}{68.7\%} of cases (paired one-sided Wilcoxon signed-rank \(p=\modelparam{results.matched.wilcoxon_less_p}{6.58\times10^{-91}}\)). Both models remain finite over all matched rollouts and preserve decreasing total number empirically.

Relative to raw PartMC, both models inherit an approximately \(\modelparam{results.mass.multistep.initial_bias_mean_pct}{-2.4\%}\)
decoded total-mass offset at initialization; this absolute representation bias
is distinct from the self-normalized temporal retention in
Figure~\ref{fig:total_mass_retention}. Over 24 h, the multistep model has mean
mass drift \(\modelparam{results.mass.multistep.drift_24_mean_pct}{-0.009\%}\), median absolute drift \(\modelparam{results.mass.multistep.drift_24_abs_median_pct}{3.3\%}\), and 95th-percentile
absolute drift of approximately \(\modelparam{results.mass.multistep.drift_24_abs_p95_pct}{18.3\%}\); the single-step values are
\(\modelparam{results.mass.single_step.drift_24_mean_pct}{-0.59\%}\), \(\modelparam{results.mass.single_step.drift_24_abs_median_pct}{4.1\%}\), and approximately \(\modelparam{results.mass.single_step.drift_24_abs_p95_pct}{20.2\%}\). These values distinguish static representation bias from subsequent rollout drift; exact mass conservation is absent from both training objectives.
Figure~\ref{fig:total_mass_retention} extends the final-model drift diagnostic
to 48 h.

\end{document}